\documentclass[final,onefignum,onetabnum]{siamart171218}

\usepackage{lipsum}
\usepackage{amsfonts}
\usepackage{graphicx}
\usepackage{epstopdf}
\usepackage{algorithmic}
\ifpdf
  \DeclareGraphicsExtensions{.eps,.pdf,.png,.jpg}
\else
  \DeclareGraphicsExtensions{.eps}
\fi

\newsiamremark{remark}{Remark}
\newsiamremark{hypothesis}{Hypothesis}
\crefname{hypothesis}{Hypothesis}{Hypotheses}
\newsiamthm{claim}{Claim}

\headers{Edge-Flipping Distance Proof Revisited}{T. Dag\`es and A. M. Bruckstein}

\title{A Bound on the Edge-Flipping Distance between Triangulations (Revisiting the Proof)}

\author{Thomas Dag\`es\thanks{Department of Computer Science, Technion Israel Institute of Technology, Haifa, ISRAEL (\email{thomas.dages@cs.technion.ac.il}).}
\and Alfred M. Bruckstein\footnotemark[1]}

\usepackage{amsopn}
\DeclareMathOperator*{\argmax}{argmax}

\usepackage[caption=false]{subfig}

\makeatletter
\newcommand*{\addFileDependency}[1]{
  \typeout{(#1)}
  \@addtofilelist{#1}
  \IfFileExists{#1}{}{\typeout{No file #1.}}
}
\makeatother

\ifpdf
\hypersetup{
  pdftitle={A Bound on the Edge-Flipping Distance between Triangulations (Revisiting the Proof)},
  pdfauthor={T. Dag\`es and A. M. Bruckstein}
}
\fi

\begin{document}

\maketitle

\begin{abstract}
  We revisit here a fundamental result on planar triangulations, namely that the flip distance between two triangulations is upper-bounded by the number of proper intersections between their straight-segment edges. We provide a complete and detailed proof of this result in a slightly generalised setting using a case-based analysis that fills several gaps left by previous proofs of the result.
\end{abstract}



\section{Introduction}

Triangulations of finite sets of points in the plane are important in many areas of applied geometry, from computer graphics and rendering, computational geometry to imaging, computational fluid dynamics, and finite element theory. However, triangulations of a same finite set of points are not unique in general, and we therefore need tools to compare different triangulations.  One tool to do this is to define a local operation called edge-flipping and define an ``edge-flipping distance'' between different triangulations of the same set of points. An edge flip consists in an operation on a triangulation, where given two adjacent triangles that form a convex quadrilateral, we flip the existing diametrical edge to the other diameter of the quadrilateral, as illustrated in \cref{fig: flip conv and non conv}. This operation maps the quadrilateral to itself but its interior is triangulated differently, the rest of the triangulation remaining unchanged. A flip does not create any conflict between the new edge and the rest of the triangulation's edges. Thus, the edge-flipping operation maps a triangulation to an other slightly different triangulation.

\begin{figure}[tbhp]
    \centering
    \subfloat[][Convex case]{
    \includegraphics[width=0.66\textwidth]{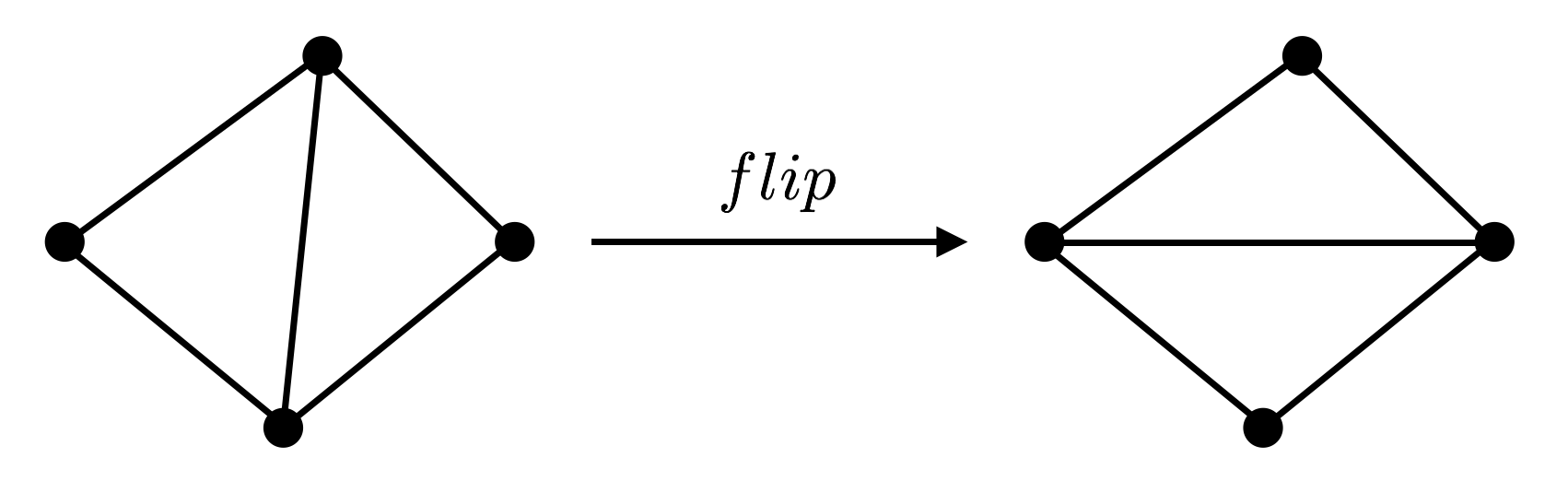}}
    \hspace{2em}
    \subfloat[][Non convex case]{
    \includegraphics[width=0.2\textwidth]{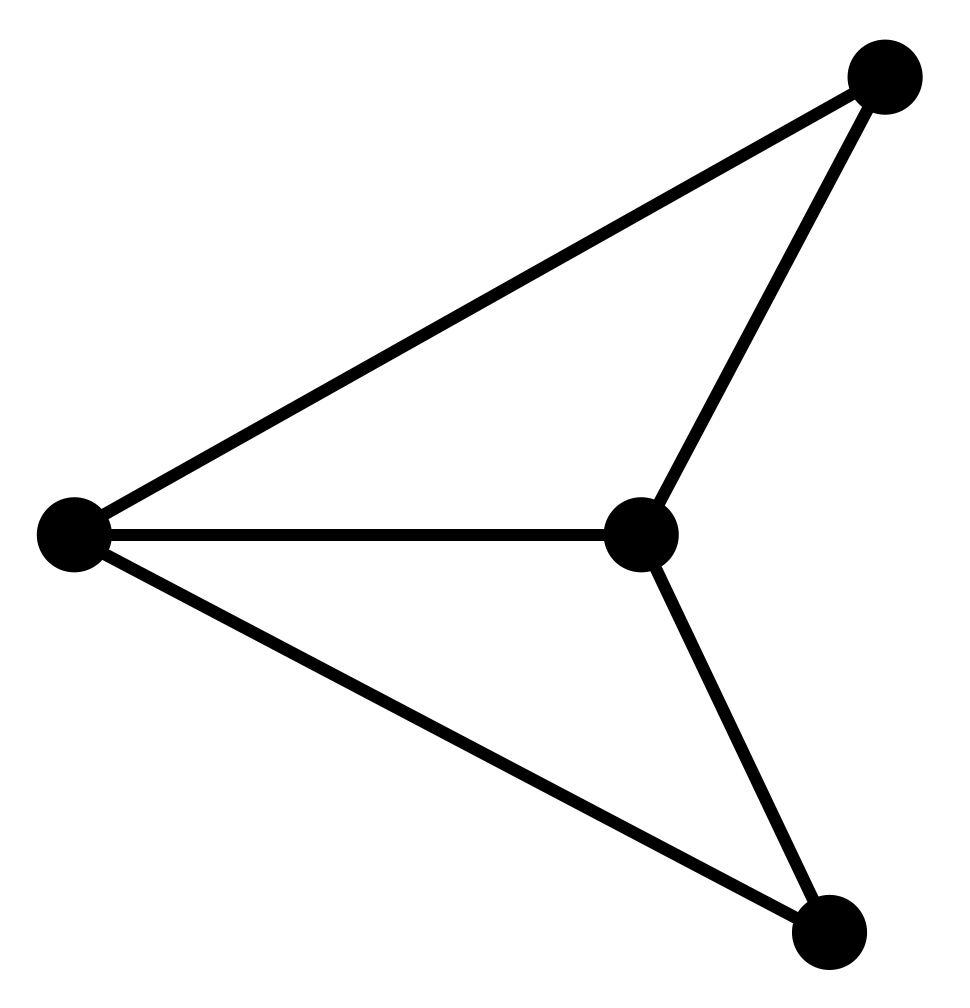}}\\ 
    \caption{Flip operation on two neighbouring triangles of a triangulation. The flip is well defined if the quadrilateral is convex. The diagonal of a non convex quadrilateral in a triangulation is not flippable.}
    \label{fig: flip conv and non conv}
\end{figure}

It is a well-known fact that given any finite set of points $S$ on the plane, and two triangulations $T_1$ and $T_2$ of this set, then there always exists a finite sequence of edge-flipping operations such that we can transform $T_1$ to $T_2$: $T_1 \xrightarrow{flip} T^{(2)} \xrightarrow{flip} T^{(3)}  \hdots \xrightarrow{flip} T_2$. The traditional proof consists in first showing that given any triangulation of planar points, we can reach a reference triangulation, and since the flip operation is clearly reversible, the result follows immediately. The Delaunay triangulation is commonly chosen as the reference triangulation, and the flip sequence to reach it is one that increases the minimum angle of the triangulation at each flip \cite{LAWSON1977flip}. Unfortunately, if there are four co-circular points, the Delaunay triangulation is not unique and extra effort is needed to prove that any Delaunay triangulation is reachable from any other one using only edge flips.  A more elegant proof that does not require any high level knowledge on triangulations, such as the Delaunay triangulation and its angular properties, can be found in Osherovich and Bruckstein \cite{osherovich2008all}. The core idea of this proof is to look at ``ears of triangulations'' of polygons and try to cut them in a goal-oriented way and then generalise by induction. 

Given two triangulations $T_1$ and $T_2$ of the same finite set of vertices $S$, there exist in general many sequences of edge flips that map $T_1$ to $T_2$. Denote $d_f(T_1,T_2)$ the shortest length of all the possible such sequences of flips. It is easy to check that $d_f$ is a distance measure over the set of triangulations of a given finite point set (for details see \cref{prop: d_f distance measure}), and it is therefore called the edge-flipping distance.

Computing the minimal sequence of flips is a difficult computational challenge. Lubiw et al. \cite{lubiw2015flip} have shown that it is actually a NP-complete task. The same year, Aichholzer et al. \cite{aichholzer2015flip} proved that the problem is also NP-complete when considering triangulations of simple planar polygons. Thus, we may search instead for an upper-bound of the edge-flipping distance. For this, we will look at edge intersections of the original and target triangulations. Due to the maximality of triangulations, we have the following equivalence: two triangulations $T_1$ and $T_2$ of a same finite set of points $S$ are identical if and only if there are no proper intersections between edges of $T_1$ with those of $T_2$ (see \cref{prop: = equiv no inter}). We say that two edges of two different triangulations intersect properly if the ``open edges'' (excluding the border vertices) intersect at a single point. Note that identical edges in two different triangulations do not properly intersect, they are superimposed.

Magically, it was found that the flip distance $d_f(T_1,T_2)$, i.e. the minimum number of flips required to go from $T_1$ to $T_2$, is upper-bounded by the number of proper intersections between the two triangulations, denoted as $\#(T_1,T_2)$. The proof of this beautiful result was originally proposed by Hanke et al. \cite{hanke1996edge} and was later revisited in the reference book on triangulations by De Loera et al. \cite{de2010triangulations}. Both proofs are quite difficult to follow due to the case based analysis involved and to some omissions, minor errors, and even some logical flaws in the proofs provided in those references. The purpose of this paper is to revisit the proof in a way that is hopefully readable, complete, and understandable by anyone with some background in planar geometry and triangulations.

\section{Edge-flipping distance and intersections: an upper-bound}
\label{sec: edge flipping dist and inter upper bound}

We consider triangulations of a fixed finite set of points $S$ in the plane. We assume that the edges of the triangulations are straight segments,  this being an essential assumption for counting proper edge intersections. In particular, we do not triangulate the outer face, i.e. the outside of the convex hull of $S$. While Hanke et al. \cite{hanke1996edge} and De Loera \cite{de2010triangulations} both consider triangulations of the convex hull of $S$ without any further constraint, we notice that the result and proof still hold for triangulations with more general border constraints, e.g. a fixed non convex outer border and some fixed inner faces called holes that are not triangulated. Formally, let $B = (B_0, B_1,\cdots, B_h)$ with $h\ge 0$ be the border constraints of $S$. Each constraint $B_i$ is a simple polygon defined on vertices of $S$. The polygons $B_i$ must not overlap, although they are allowed to share vertices but not edges. $B_0$ is taken to be a simple polygon of $S$ such that its outer face does not contain any point of $S$. $B_0$ may be chosen as the boundary of the convex hull of $S$, but we will not limit our choice to this example. The set of further polygonal constraints $(B_1,\cdots,B_h)$, which can be empty, defines holes of the triangulation. Each $B_i$ for $i\ge 1$ is defined such that their inner region does not contain any point of $S$. Note that the holes are not necessarily triangles and can have arbitrarily long polygonal boundaries $B_i$. See \cref{fig: border example} for an illustration of triangulation with general border constraints.

\begin{figure}[tbhp]
    \centering
    \includegraphics[width=0.7\textwidth]{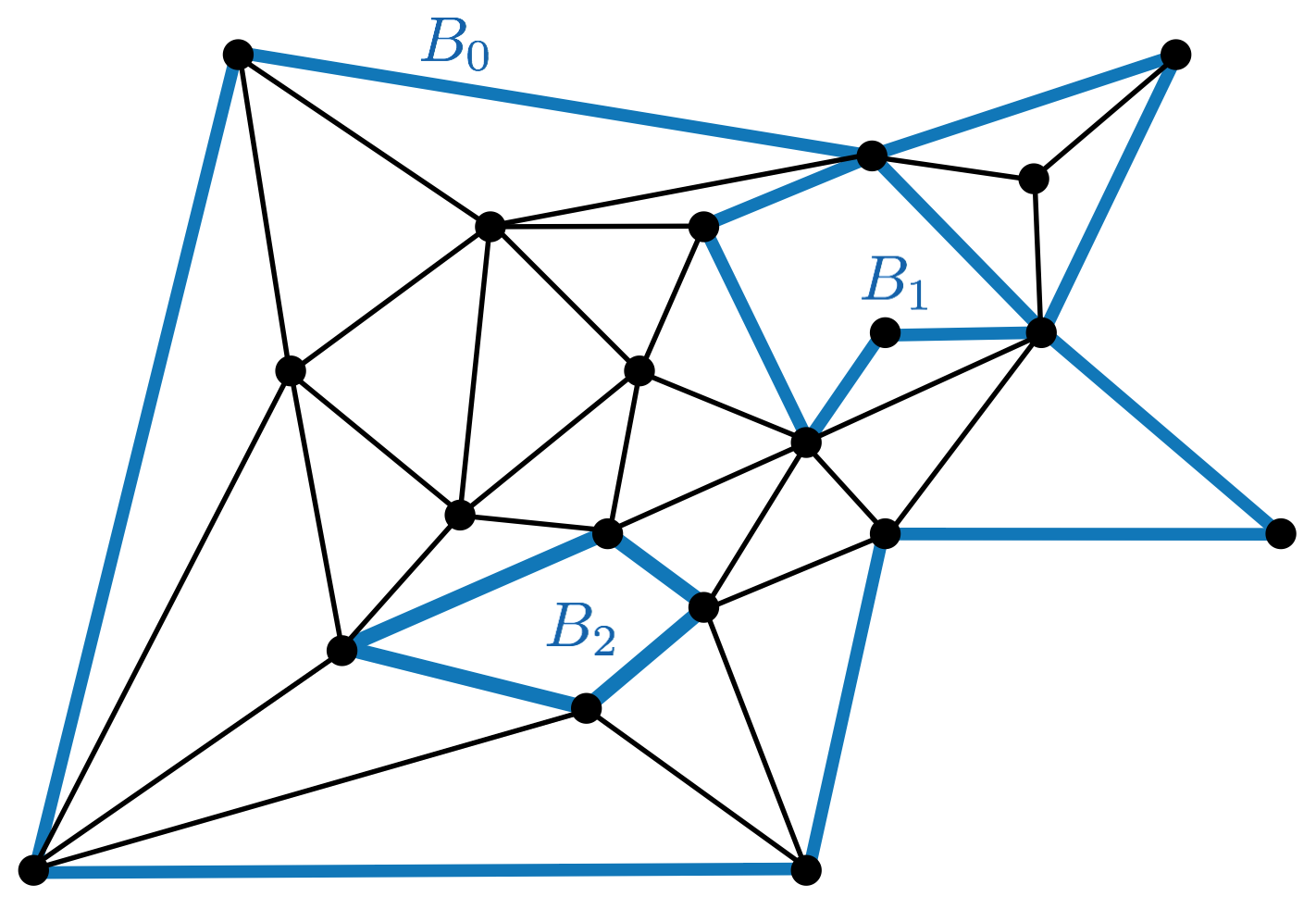}
    \caption{Example of a triangulation with general border constraints $B$. Here, we constrain the triangulation to have $h=2$ fixed holes defined by $B_1$ and $B_2$. The outer border of the triangulation is constrained to the polygon $B_0$, which in this case is not convex as it is not the polygon defined by the border of the convex hull of $S$. Given $B = (B_0, B_1, B_2)$, we will compare different triangulations of $S$ that respect the border constraints $B$.}
    \label{fig: border example}
\end{figure}

We want to show that the edge-flipping distance of points in any such region of the plane is upper-bounded by the number of edge intersections between both triangulations. The proof is heavily based on the original ideas of \cite{hanke1996edge} and \cite{de2010triangulations} with all details thoroughly explained, and hopefully all problems removed. The goal is to prove the following theorem:

\begin{theorem}
\label{th: goal}
    If $T_1$ and $T_2$ are two triangulations of a planar region defined by a finite set of points $S$ with a set of boundary constraints $B$, then we have that:
    $$ d_f(T_1,T_2) \le \#(T_1,T_2). $$
    In particular, if $n$ is the number of vertices in $S$, if $n_b$ is the number of boundary vertices of $S$, and if $h$ is the number of holes of $S$, i.e. $B = (B_0, \cdots, B_h)$ and $n_b$ is the sum of the length of the polygons $B_i$ for $0\le i \le h$, then:
    $$ d_f(T_1,T_2) \le \#(T_1,T_2) \le (3n-2n_b-3-3h)^2.$$
\end{theorem}

The main idea for the proof is to prove the existence a sequence of edge flips that strictly decreases the number of intersections at each step. To achieve this, we consider the edges in $T_1$ with maximal number of intersections with $T_2$. We then prove the existence of one of these maximal edges such that it is flippable and that flipping it will strictly reduce the number of intersections with $T_2$. We will first show that such maximal edges are flippable, which means that they are diagonals of convex quadrilaterals in $T_1$. Second, we will show that there exists a class of configurations of maximal edges for which flipping strictly reduces the number of intersections. Third, we will show that if no maximal edge belongs to this class of configurations, then there must exist at least one maximal edge that does not belong to this class which will strictly reduce the number of intersections with $T_2$ after being flipped as we would otherwise reach a contradiction. Therefore, there is at least one edge in $T_1$ that we can flip and for which this flip strictly reduces the number of intersections with $T_2$. This result proves \cref{th: goal}.

We will denote, unless mentioned otherwise, the edge between the vertices $a$ and $b$ as $ab$. We will call a quadrilateral of a triangulation $T$ the shape formed by two adjacent triangles of $T$ (sharing an edge). Given a quadrilateral created by two adjacent triangles, we will say that an edge intersects with this quadrilateral if it crosses the interior of the quadrilateral. If $ab$, $cd$ and $ef$ are edges between vertices of the triangulations, we will denote $\#(ab,T_2)$ the number of edges intersecting $ab$ that belong to $T_2$, $\#(ab,cd,T_2)$ the number of edges intersecting $ab$ and $cd$ that belong to $T_2$, $\#(ab,cd,ef,T_2)$ the number of edges intersecting $ab$, $cd$ and $ef$ that belong to $T_2$, $\#_a(cd,T_2)$ the number of edges intersecting $cd$ that belong to $T_2$ but that emerge from $a$ and finally $\#_a(cd,ef,T_2)$ the number of edges intersecting $cd$ and $ef$ that belong to $T_2$ but that emerge from $a$, where an edge is said to emerge from $a$ if $a$ is one of its vertices. See figure \cref{fig: notations inter} for an illustration of each case.

\begin{figure}[tbhp]
    \centering
    \subfloat[][$\#(ab,T_2)$]{
    \includegraphics[width=0.25\textwidth]{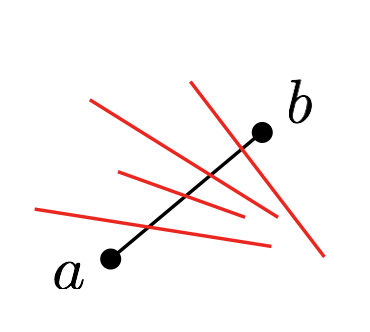}}
    \subfloat[][$\#(ab,cd,T_2)$]{
    \includegraphics[width=0.35\textwidth]{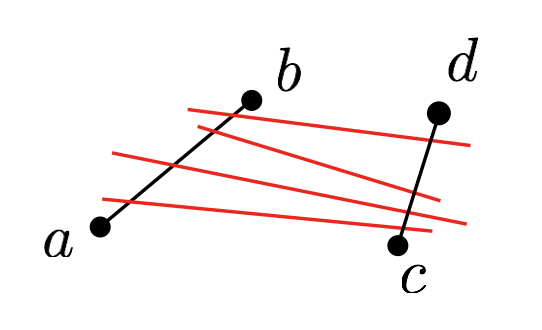}}
    \subfloat[][$\#(ab,cd,ef,T_2)$]{
    \includegraphics[width=0.4\textwidth]{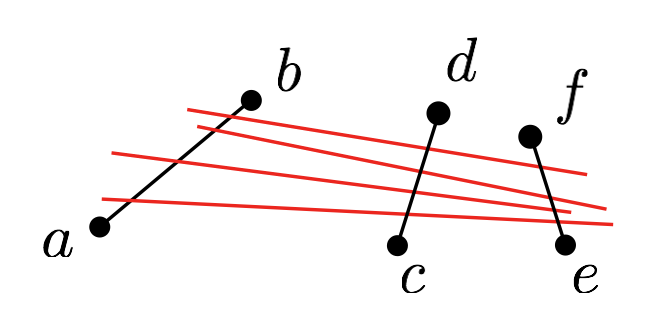}}\\
    \subfloat[][$\#_a(bc,T_2)$]{
    \includegraphics[width=0.3\textwidth]{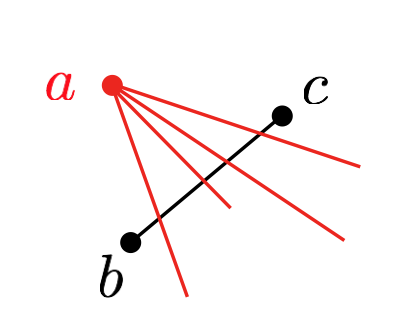}}
    \subfloat[][$\#_a(bc,de,T_2)$]{
    \includegraphics[width=0.35\textwidth]{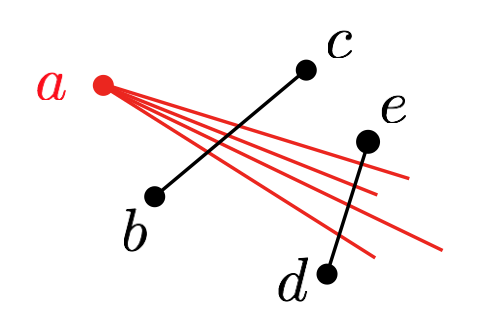}}
    \caption{Notations for counting the types of intersecting edges. The edges drawn in black are assumed to belong to $T_1$ and those in red are assumed to belong in $T_2$. This colour code will remain consistent throughout this paper.}
    \label{fig: notations inter}
\end{figure}

Note that while Hanke et al. \cite{hanke1996edge} stated \cref{th: goal} with a strict inequality, it is easy to find equality cases. For instance, any triangulation of three points is trivially reduced to a triangle and there cannot be any intersection. Another simple example consists in the two possible triangulations of a convex set of four points, which differ by only one flipped edge and have exactly one edge intersection. This small misstatement was also corrected in the proof of De Loera et al. \cite{de2010triangulations}.

\subsection{Preliminary properties of triangulations}

Some of the important tools for proving theorem \cref{th: goal} are listed as the following properties:

\begin{proposition}
\label{prop: planar}
    If $T$ is a triangulation of a finite set of points, then $T$ is a  planar graph, i.e. its edges do not intersect.
\end{proposition}

\begin{proposition}
\label{prop: no vertex inside quad}
    Let $T_1$ and $T_2$ be two triangulations of the same finite set of points. Consider any quadrilateral of $T_1$ defined by two adjacent triangles of $T_1$. Consider any edge of $T_2$ intersecting with the quadrilateral. Necessarily, the considered edge is one of the three possible kinds: both vertices of the edge are outside the quadrilateral, or one of the vertices is one of the vertices of the polygon and the other vertex is outside the polygon, or the edge is a diagonal of the quadrilateral.
\end{proposition}

\begin{proposition}
\label{prop: no intersec inside quad}
    Let $T_1$ and $T_2$ be two triangulations of the same finite set of points. Consider any quadrilateral of $T_1$ defined by two adjacent triangles. Consider any two edges of $T_2$ that both intersect the quadrilateral. Then these edges cannot meet inside the quadrilateral.
\end{proposition}

\begin{proposition}
\label{prop: vertex linked to the closest edge}
    Let $T$ be a triangulation of a finite set of points $S$. Let $a\in S$ be a vertex of $T$ and $x\in\mathbb{R}^2$ be a point that is not necessarily in $S$. Assume that the segment $[ax[$ open at $x$ linking $a$ to $x$ does not intersect with any edge of $T$. Then if $x$ is a vertex of $T$, i.e. $x\in S$, then $ax$ is an edge of $T$. Likewise, if $x$ is not a vertex of $T$ but is on an edge $bc$ of $T$ with $b\neq a$ and $c\neq a$, then $ab$ and $ac$ are edges of $T$. 
\end{proposition}

\begin{proposition}
\label{prop: = equiv no inter}
    Let $T_1$ and $T_2$ be two triangulations of a same finite set of points, then $T_1\neq T_2$ if and only if the number of intersections between both triangulations is strictly positive. 
\end{proposition}

\begin{proposition}
\label{prop: d_f distance measure}
    The function $d_f$ between triangulations of a same finite set of points is a distance measure.
\end{proposition}

\begin{proposition}
\label{prop: inter edge not edge other triangu}
    Let $T_1$ and $T_2$ be two triangulations of a same finite set of points. If $T_2$ intersects an edge $ac$ of $T_1$, then $ac$ is not an edge of $T_2$. In particular, when $T_1\neq T_2$, the edges $ac\in \argmax\limits_{\widetilde{ac}\in T_1} \#(\widetilde{ac},T_2)$ of $T_1$ with maximal number of intersections with $T_2$ cannot be edges of $T_2$.
\end{proposition}

\begin{proposition}
\label{prop: inter not on border}
    Let $T_1$ and $T_2$ be two triangulations of a same finite set of points. If $ac$ an edge of $T_1$ with at least one intersection with $T_2$, then $ac$ cannot be a border edge. Likewise, all border edges are in both $T_1$ and $T_2$.
\end{proposition}

\begin{proof}
    \Cref{prop: planar} follows from the definition of a triangulation. Triangulations are defined as maximal planar graphs, i.e. such that adding any extra edge to the triangulation would render it non planar and thus cause an intersection. One can see that if all the faces are not triangles then there is a face with 4 or more vertices. By definition, the inside of the face is empty without any edge passing through it. Therefore, we can simply add an edge between two non adjacent vertices of the face passing through the face in order to violate the maximality assumption while still being a planar graph. Reciprocally, if we have only triangles for faces, in order to be able to add an edge without causing intersections this would mean that there is a face with an available diagonal at which we can add an edge. But faces with a diagonal have at least 4 vertices, which contradicts the triangles only assumption. Thus the graph is maximally planar. Therefore, we have the equivalence of triangulations as maximal planar graphs and as connected graphs that have only triangle faces. What is important to remember is that should two edges of a triangulation meet, then necessarily they meet at a vertex of the triangulation.
\end{proof}

\begin{proof}
    \Cref{prop: no vertex inside quad} is a consequence of the fact that $T_1$ and $T_2$ are the triangulation of the same set of points $S$. As such, when defining the quadrilateral as two adjacent triangles in $T_1$, there cannot be any other vertex of $S$ inside the quadrilateral as it would be inside one of the faces or inside the diagonal edge (not on the extremities). Therefore, if an edge of $T_2$ passes inside the quadrilateral, then necessarily its vertices are either outside the quadrilateral or on its boundary. The only possibilities for a vertex being on the boundary is for it to be one of the four vertices of the quadrilateral since an edge of a triangulation cannot pass through a vertex. By definition, the edges of the boundary of the quadrilateral do not intersect the polygon as they do not pass through its interior. Thus the three possibilities: the edge in $T_2$ is a diagonal of the quadrilateral, or one of its vertices is on the polygon but the other one is outside it, or both vertices are outside it. In summary, there cannot be a vertex inside the polygon.
\end{proof}
    
\begin{proof}
    \begin{figure}[tbhp]
        \centering
        \includegraphics[width=0.3\textwidth]{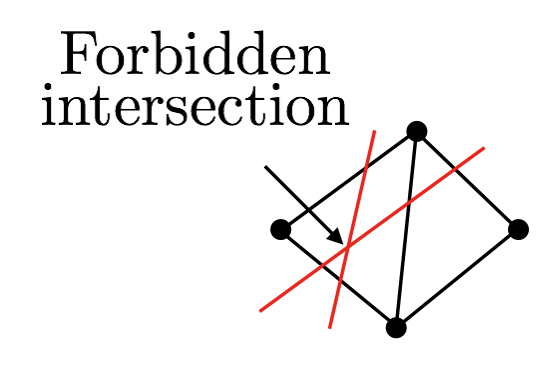}
        \caption{Edges of $T_2$ cannot meet inside a quadrilateral composed of two neighbouring triangles of $T_1$.}
        \label{fig: no intersec inside quad}
    \end{figure}
    \Cref{prop: no intersec inside quad} is proven similarly to the  previous property. See \cref{fig: no intersec inside quad} for an illustration. Two edges of $T_2$ cannot intersect since $T_2$ is planar (\cref{prop: planar}), they are either disjoint or they meet at a common vertex. Now assume two edges of $T_2$ intersect a quadrilateral of adjacent triangles of $T_1$. If these edges meet inside the quadrilateral, then there is a vertex inside it. However, since $T_1$ and $T_2$ triangulate the same set of vertices, that means that there is a vertex inside the triangulated quadrilateral, which is a contradiction. Therefore, edges of $T_2$ cannot meet inside the quadrilateral.
\end{proof}

\begin{proof}
    \begin{figure}[tbhp]
        \centering
        \includegraphics[width=0.8\textwidth]{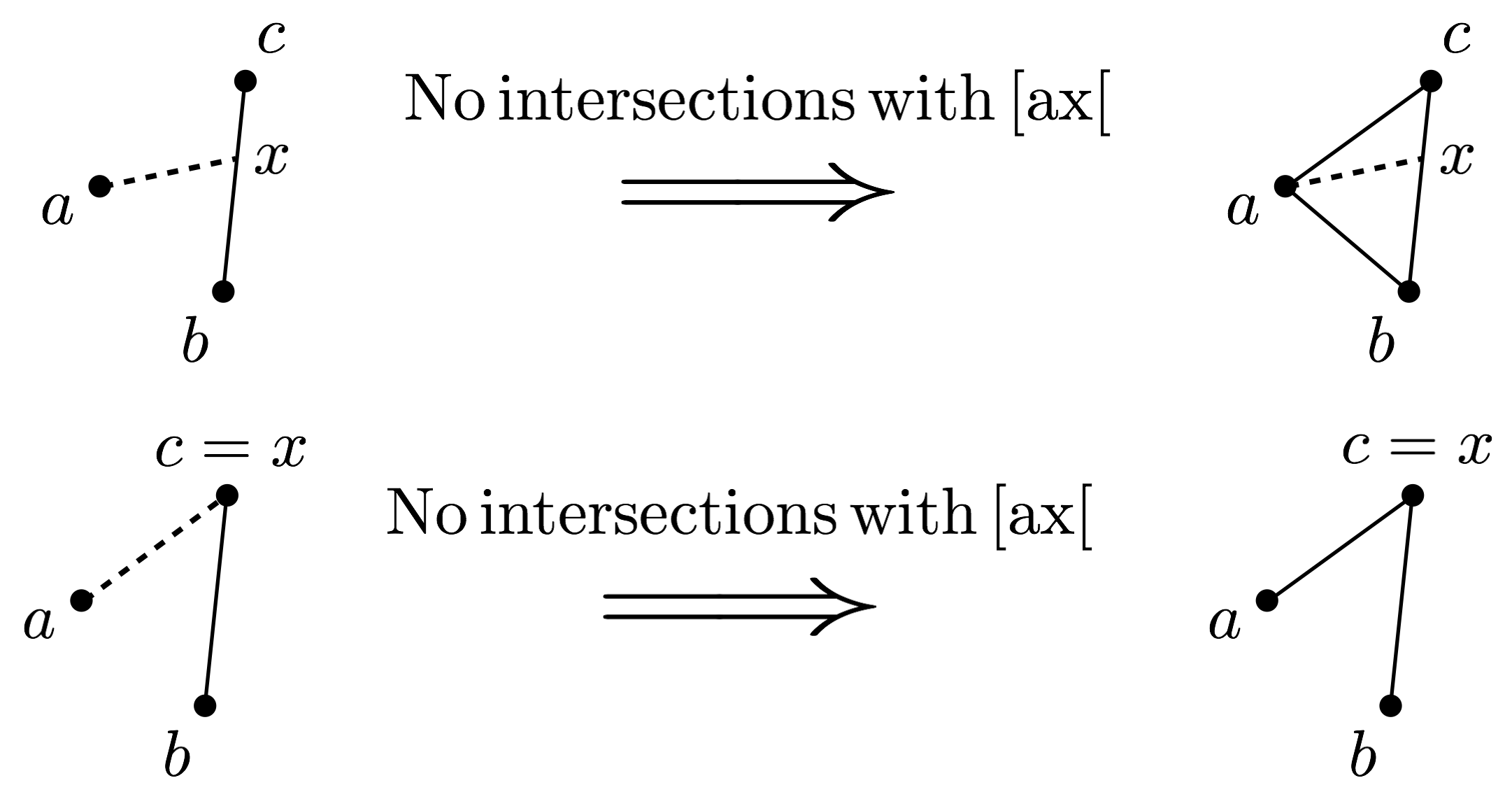}
        \caption{If $T$ has no intersection with $[ax[$ and $x$ is on an edge of $T$ without being a vertex, then $ab$ and $bc$ are edges of $T$. If instead $x$ is a vertex of $T$, then $ax$ is an edge of $T$.}
        \label{fig: vertex linked to the closest edge}
    \end{figure}
    \Cref{prop: vertex linked to the closest edge} is one of the most important tools in proving \cref{th: goal}. Both Hanke et al. \cite{hanke1996edge} and De Loera et al. \cite{de2010triangulations} implicitly use it, however they neither explicitly formulate it, nor prove it. The proof follows from the fact that triangles only have three edges. See  \cref{fig: vertex linked to the closest edge} for an illustration of both cases. Consider the case when $x$ is not a vertex of $T$ but is on an edge $bc$ of $T$ with $a\notin\{b,c\}$. The open segment $]ax[$ does not intersect with any other edge of the triangulation and is not included in one of the edges. As such, it is necessarily included in a face $F$. Since $]ax[\subset F$, $a$ is a vertex of $F$ and $bc$ is an edge of $F$. In turn, this implies that the vertices of $F$ are $a$, $b$, and $c$. Therefore, $ab$ and $ac$ are edges of $T$. If $x$ is now a vertex of $T$, then $ax$ has no intersection with any other edge of $T$. If $ax$ is not an edge of $T$, then adding it to $T$ would not create an intersection, which would violate the maximality principle of the triangulation $T$. Therefore, $ax$ is necessarily part of the triangulation.
\end{proof}

\begin{proof}
    \Cref{prop: = equiv no inter} is the fundamental property for the edge-flipping distance to be a proper distance. If the triangulations are identical, then there are no intersections. Reciprocically, assume that the triangulations have no intersections. If they are not equal, then one of the triangulations has an edge that is not in the other. Without loss of generality assume $ab$ is an edge of $T_1$ but not of $T_2$. Since $ab$ does not intersect any edge of $T_2$ then the open edge $]ab[$ is included in the interior of a face of $T_2$, and thus it is a diagonal of a face of $T_2$. Adding the edge $ab$ to $T_2$ would then produce an expanded planar graph, which violates the maximality of the triangulation $T_2$. Therefore, both triangulations share the same edges, implying that they are equal.
\end{proof}

\begin{proof}
    \Cref{prop: d_f distance measure} can be proven by simply checking that $d_f$ satisfies the defining properties of distance measures. The measure $d_f(T_1,T_2)$ is positive. Due to the fact that edge flips are invertible, with inverses being edge flips, the sequence of opposite flips from $T_2$ to $T_1$ is of minimal length. This result gives the symmetry $d_f$. If the two triangulations are identical then no edge flip is necessary and  if no edge flip is needed then the triangulations are identical (see \cref{prop: = equiv no inter}). The minimality provides the triangular inequality. If $T_1$, $T_2$, and $T_3$ are three triangulations of $S$, then the shortest sequence of edge flips from $T_1$ to $T_3$ does not necessarily pass through $T_2$ and is necessarily shorter than any sequence going from $T_1$ to $T_2$ and then from $T_2$ to $T_3$. In particular, it is shorter than the sum of minimal ones. Thus $d_f$ is a distance measure.
\end{proof}

\begin{proof}
    The proof of \cref{prop: inter edge not edge other triangu} directly follows from the planarity of triangulations. Indeed, assume an edge $ac$ of $T_1$ has an intersection with an edge $bd$ of $T_2$. If $ac$ is an edge of $T_2$, then $ac$ and $bd$ are two edges of $T_2$ that intersect, which is impossible since $T_2$ is planar. From \cref{prop: = equiv no inter}, if $T_1\neq T_2$, then they have least one intersection. This implies that the maximally intersecting edges of $T_1$ with $T_2$ are not edges of $T_2$.
\end{proof}

\begin{proof}
    \Cref{prop: inter not on border} is a consequence of \cref{prop: inter edge not edge other triangu}. Indeed, consider the border constraint $B_i$ for $0\le i\le h$. We can define a cyclical ordering on the vertices on the border polygon $B_i$, such as a clockwise ordering. Then, each border vertex is linked to its next neighbour, otherwise adding that edge would preserve planarity while respecting the border constraints (as border polygons do not overlap nor share edges) but violate the maximality of triangulations. Therefore, all triangulations of the same finite set of points share the same border edges. Since triangulations are planar, i.e. they do not have intersecting edges, and since the edges are taken to be straight and thus cannot self-intersect, then intersecting edges between triangulations cannot be border edges.
\end{proof}

\subsection{Upper-bounding the flip distance}

We are now equipped with the basic tools for proving \cref{th: goal}. In all the following, we will assume that $T_1$ and $T_2$ are two triangulations of a same finite set of points $S$ with $T_1\neq T_2$. Due to \cref{prop: = equiv no inter} there is at least one intersection. In particular, if $ac$ an edge with maximal number of intersections with $T_2$, then due to  \cref{prop: inter edge not edge other triangu} $ac$ cannot be an edge of $T_2$.

Following the sketch of the proof of this theorem outlined in the introduction of \cref{sec: edge flipping dist and inter upper bound}, the first step is to prove that all edges with maximal intersections can be flipped.

\begin{lemma}
\label{lemma 1}
    All quadrilaterals of $T_1$, formed by adjacent triangles of $T_1$, that contain a diagonal with maximum number of intersection with $T_2$ are necessarily convex. 
\end{lemma}

\begin{proof}
    First, realise that if the quadrilateral is not convex, then by definition the flip of the diagonal is an impossible operation (see \cref{fig: flip conv and non conv}). Let us prove that the quadrilateral $abcd$ with diagonal $ac$ in $T_1$ is convex, where $ac$ has maximum number of intersections with $T_2$. Assume the quadrilateral is not convex and that $a$ is the reflex vertex, i.e. the outer angle $\widehat{bad}<\pi$. The trick will be to prove that all edges that intersect $ac$ will necessarily also intersect $cd$ and $cb$, which will imply by maximality that $cd$ and $cb$ also have maximal intersections with $T_2$. And then we will find an edge that intersects $cd$ or $cb$ but not $ac$ and this will violate the maximality assumption of $ac$.

    \begin{figure}[tbhp]
        \centering
        \includegraphics[width=0.3\textwidth]{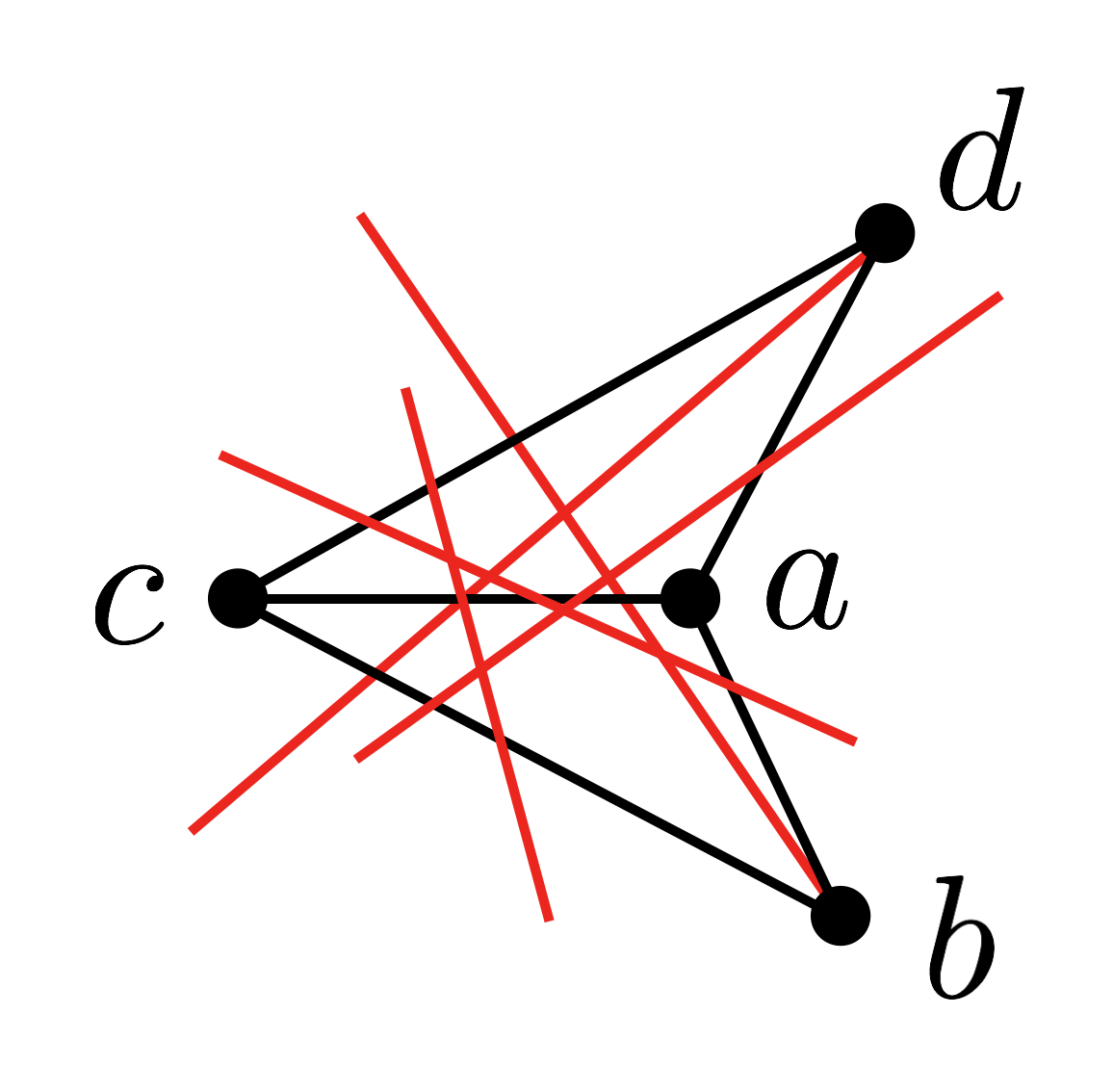}
        \caption{All 5 possible types of edges of $T_2$ that intersect $ac$ in a non convex configuration.}
        \label{fig: lemma 1 5 types intersec}
    \end{figure}
    
    There are $5$ types of possible edges that can intersect $ac$ (see \cref{fig: lemma 1 5 types intersec}): those that intersect $cb$ and $cd$ while intersecting $ac$, those that intersect $cd$ and come from $b$ while intersecting $ac$, those that intersect $ab$ and $cd$ while intersecting $ac$, those that intersect $ad$ and $cb$ while intersecting $ac$ and those that come from $d$ and intersect $cb$ and $ac$. This can be summarised in the following summation where it is important to remember that different terms ine the sum correspond to different edges:
    
    \begin{align}
    \label{ac,T2 lemma 1}
        \#(ac,T_2) = \, &\#(ab,cd,ac,T_2) + \#(bc,cd,ac,T_2) + \#(da,bc,ac,T_2) \nonumber\\
        &+ \#_b(cd,ac,T_2) + \#_d(bc,ac,T_2)
    \end{align}

    \begin{figure}[tbhp]
        \centering
        \includegraphics[width=0.6\textwidth]{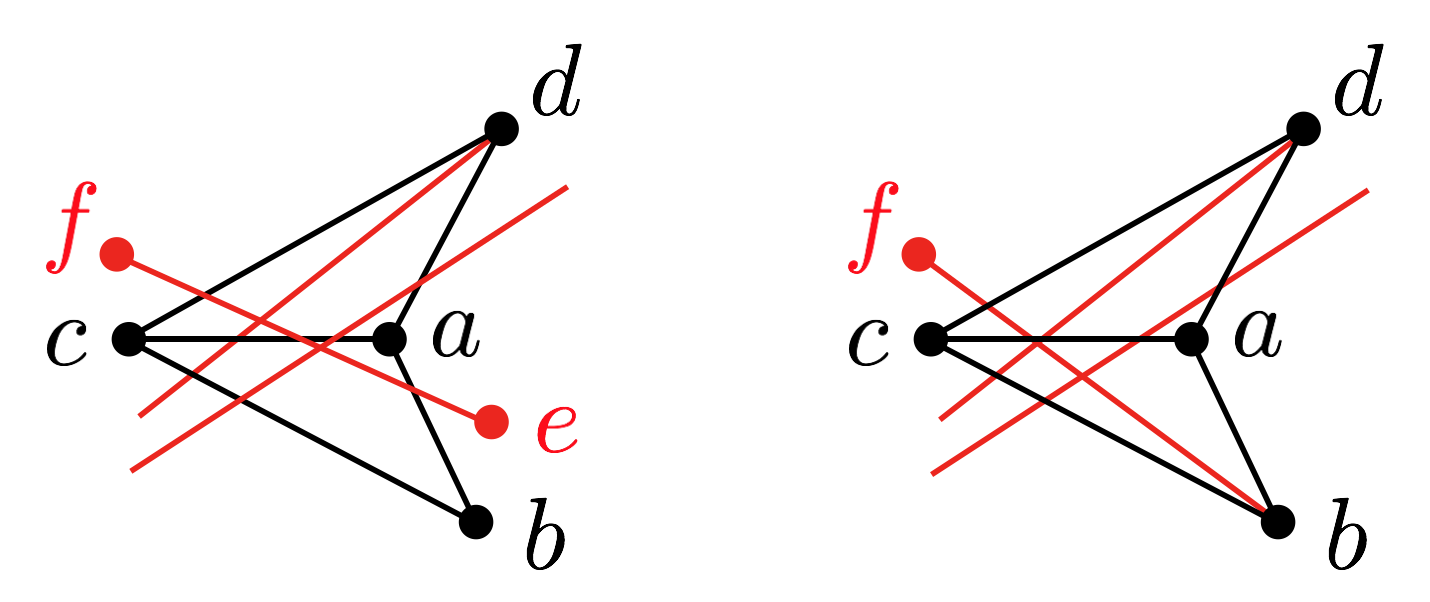}
        \caption{If $T_2$ has an edge $ef$ (or $bf$) intersecting $ac$ that passes through $]ab]$, then it cannot also have an edge intersecting $ac$ that passes through $]ad]$ since that would create illegal self intersections inside the non-convex quadrilateral because $a$ is a reflex vertex.}
        \label{fig: lemma 1 ef implies no ad}
    \end{figure}
    
    \begin{figure}[tbhp]
        \centering
        \includegraphics[width=0.7\textwidth]{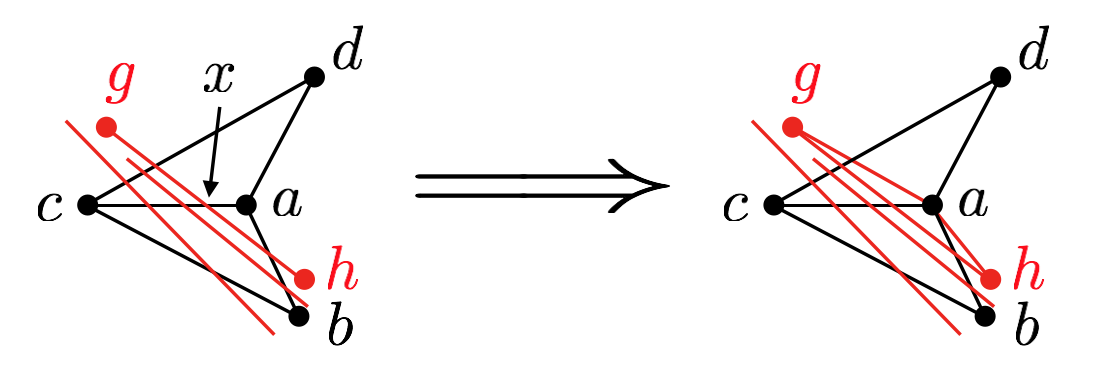}
        \caption{The closest to $a$ intersecting edge $gh$ of $T_2$ with $ac$ that passes through $]ab]$ intersects $ac$ at point $x$. In order to satisfy the triangulation assumption, necessarily $ag$ and $ah$ are in the triangulation. At least one of these two edges intersects $cd$ but they do not intersect $ac$, which violates the maximality assumption on $ac$.}
        \label{fig: lemma 1 x ag not all vertical}
    \end{figure}

    Assume that not all of them pass through $cb$ and $cd$. For instance, there is an edge $ef$ in $T_2$ that intersects $ab$ and $cd$ and that intersects $ac$, or that comes from $b$ and intersects $cd$ and $ac$ (see \cref{fig: lemma 1 ef implies no ad}). Then, since the quadrilateral is not convex, any edge that intersects $ad$ and $ac$ also intersects $ef$ inside the quadrilateral. Similarly, all edges that come from $d$ and intersect $ac$ also intersect $ef$ inside the quadrilateral. Due to \cref{prop: no intersec inside quad}, there cannot be an intersection of edges of $T_2$ inside the quadrilateral. Thus, no edge from $T_2$ passes through $]ad]$ (intersects $ad$ or comes from $d$). Therefore, all edges intersecting $ac$ necessarily intersect $cd$. We can now zero all terms that do not use $cd$ in the previous formula:
    \begin{equation}
        \#(ac,T_2) = \, \#(ab,cd,ac,T_2) + \#(bc,cd,ac,T_2) + \#_b(cd,ac,T_2).
    \end{equation}
    
    Therefore, there are at least as many edges intersecting $cd$ than $ac$. By maximality of $ac$, $ad$ is maximal and cannot have any other edge intersecting it that does not also intersect $ac$. Denote $x$ the closest intersection point on $ac$ to $a$ (see \cref{fig: lemma 1 x ag not all vertical}). This point $x$ does not belong to the vertices $S$ but lies on an edge $gh$ of $T_2$. Then $]ax[$ does not intersect any edge of $T_2$ and is not included in an edge of $T_2$ since $ac$ is not in $T_2$ because it has intersections with the triangulation (\cref{prop: inter edge not edge other triangu}). Therefore, $ag$ and $ah$ are edges of $T_2$ (\cref{prop: vertex linked to the closest edge}).

    Since $gh$ intersects $cd$, one of its vertices is on the other half-plane delimited by $cd$ from $a$. Without loss of generality choose this vertex to be $g$. Furthermore, since $gh$ also passes through the semi-open edge $]ab]$, $g$ also lies in the cone defined by $cba$ (in the direction of the face $cba$). Because $a$ is a reflex vertex, then all points inside this cone and on the opposite half-plane delimited by $cd$ are also in the cone delimited by $cad$, implying that they would create an intersection with $cd$ when linking them with $a$. In particular, this result holds for $g$. Therefore, $ag$ intersects $cd$. However, $ag$ does not intersect $ac$. This is contradictory to the previous results we established: all edges intersecting $ac$ intersect $cd$ and by maximality there are no edges intersecting $cd$ that do not intersect $ac$. We have thus proven that all edges intersecting $ac$ also intersect $cd$ and $bc$.
    
    \begin{figure}[tbhp]
        \centering
        \includegraphics[width=0.7\textwidth]{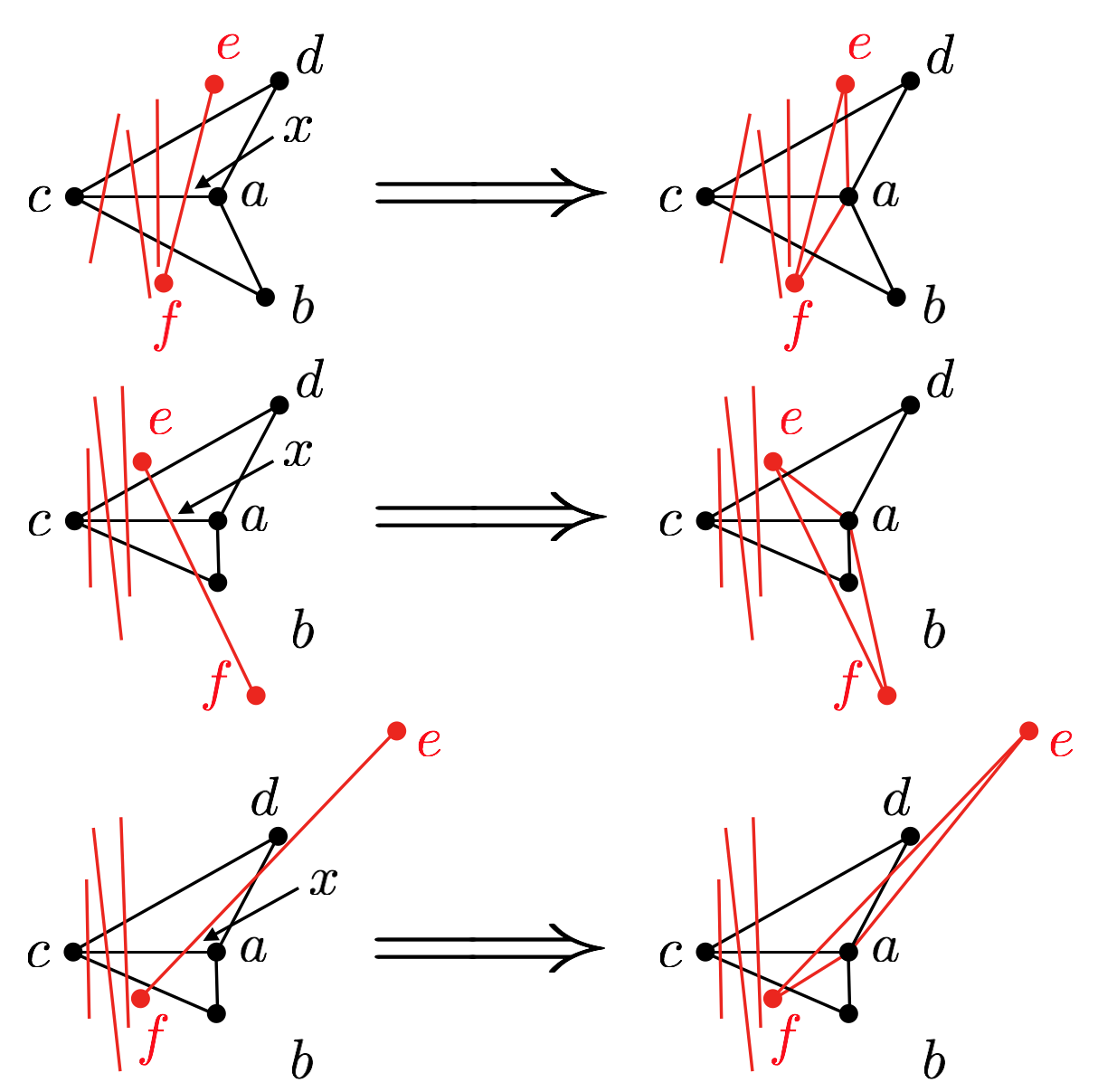}
        \caption{All edges of $T_2$ intersecting $ac$ also intersect $cd$ and $bc$, which implies that they are also maximal. Thus all edges intersecting $cd$ or $bc$ are exactly those that intersect $ac$. Consider the closest edge intersecting $ac$, at point $x$. In order to satisfy the triangulation assumption, $ae$ and $af$ are edges of $T_2$, but at least one (not necessarily both) of these edges intersects either $cd$ or $bc$ but not $ac$, which leads to a contradiction.}
        \label{fig: lemma 1 x ae all vertical}
    \end{figure}

    This result implies that the number of intersections of $cd$ with $T_2$ is at least the number of intersections of $ac$ with $T_2$. By maximality of $ac$, $cd$ is also maximal. Similarly, $cb$ is also maximal. All edges that intersect $cd$ or $cb$ must intersect $ac$. We now apply a similar reasoning to what we previously did (see \cref{fig: lemma 1 x ae all vertical}). Since the intersecting edges of $ac$ intersect $cd$ and $cb$, they cannot come from $d$ or $c$ nor pass through $ad$ or $ab$. Therefore, if we denote $ef$ an intersecting edge of $ac$, then necessarily $e$ and $f$ are different from $d$ and $b$, which implies that they both need to be outside of the quadrilateral as there are no vertices inside the quadrilateral. 
    
    Consider $ef$ the edge intersecting $ac$ at point $x$ closest to $a$. Because it is the closest to $a$, $[ax[$ has no intersection with $T_2$, with $x$ on the edge $ef$ of $T_2$. Thus $ae$ and $af$ are edges of $T_2$ (see \cref{prop: vertex linked to the closest edge}). Furthermore, because $a$ is a reflex vertex, at least one of the vertices $e$ or $f$ needs to be inside one of the cones $cad$ or $cab$, as otherwise $ef$ would not intersect $ac$. Assume without loss of generality that $e$ lies in the cone $cad$ (in the direction of the face $cad$ but on the other side from $a$ of the edge $cd$). Then $ea$, an edge of $T_2$, intersects $cd$, but does not intersect $ac$. This statement is in contradiction to the previous results: all edges of $T_2$ intersecting $cd$ (or $cb$) intersect $ac$. 
    Given the assumption that the quadrilateral $abcd$ is not convex, we have systematically reach a contradiction in all cases. Therefore $abcd$ is necessarily convex.
\end{proof}

Hanke et al. \cite{hanke1996edge} proved a weaker version of \cref{lemma 1}. Indeed, they only showed that there exists at least one convex quadrilateral with diagonal having maximal intersections. This is due to the fact that once they proved that all edges intersecting $ac$ also intersect $cd$ and $bc$, they then claimed that these two edges could not be on the border of the convex hull (or more generally on the border of the triangulation). This observation led them to reiterate their argument by induction on the next quadrilateral until either reaching a desired convex quadrilateral with diagonal having maximal intersections or reaching a contradiction as they would have to reach the border of the convex hull (or more generally the border of the triangulation). These argument are not incorrect. However, the proof of \cref{th: goal} requires the knowledge that all quadrilaterals with diagonal having maximal number of intersections must be convex. Unfortunately, Hanke et al. \cite{hanke1996edge} implicitly assume this knowledge, that was neither claimed nor proven when they call upon crucial similar inductive arguments in the rest of their proof. 

De Loera et al. \cite{de2010triangulations} reinforced the original paper's claim and fixed the last part of the proof of \cref{lemma 1}, by invoking the edge $ef$ of $T_2$ intersecting $ac$ closest to $a$, to show that all quadrilaterals with diagonal having maximal intersections are convex. However, their proof is incomplete and their illustration of the edge $ef$ does not tell the full story. Indeed, they are misled by a non general drawing that omits a possible case. Their illustration corresponds to that of the top case in \cref{fig: lemma 1 x ae all vertical}, with an incorrect claim that since $ef$ intersects $ac$ and $cd$, then $ae$ intersects $cd$. This claim is wrong as shown on the bottom of \cref{fig: lemma 1 x ae all vertical}. Indeed, merely having that $ef$ intersects $ac$, $bc$, and $cd$ only implies that at least one of $ae$ or $af$ intersects $cd$ or $bc$. Note that while this may seem like nitpicking, in a contradiction proof based on a hierarchy of case analyses, it is essential to explore all cases. In particular, while graphical illustrations provide insight, they can easily mislead by omitting cases and hence lead to false claims. While in this case the omission has only a minor impact, as one could debate that the reasoning was implicitly without loss of generality, and while one easily completes the considerations for the overlooked cases, such oversights could be much more detrimental at later stages of the proof.

Let us return to the proof \cref{th: goal}. We now know that all edges in $T_1$ with maximal intersections with $T_2$ can be flipped. Unfortunately, not all maximally intersecting edges will necessarily reduce the total number of intersections with $T_2$. However, we only need to find one of these maximal edges that reduces the number of intersections after being flipped. Next, we will scrutinise a class of configurations of quadrilaterals with maximal intersections diagonal for which we are guaranteed to strictly reduce the number of intersections with $T_2$ by flipping the diagonal edge.

\begin{lemma}
\label{lemma 2}
    Let $abcd$ be a convex quadrilateral in $T_1$ with diagonal $ac$ with maximum number of intersections with $T_2$. If $T_2$ contains an edge $eb$ intersecting $da$ or $cd$, or if it contains and edge $dg$ intersecting $ab$ or $bc$, or if $bd$ an edge of $T_2$, then flipping $ac \xrightarrow{flip} bd$ reduces the total number of intersections with $T_2$. In other words, if an edge of $T_2$ intersects a maximal edge of $T_1$ and comes from a vertex of the quadrilateral around the maximal edge, then flipping that maximal edge reduces the number of intersections.
\end{lemma}

\begin{proof}
    First, look at the easy case where $bd$ is an edge of $T_2$. As $bd$ belongs to $T_2$, then it has no intersections with this triangulation. On the other hand, since $T_1\neq T_2$ and $ac$ has maximal number of intersections with $T_2$, then $ac$ has at least one intersection with $T_2$. Therefore flipping $ac$ to $bd$ reduces the number of intersections by $\#(ac,T_2) - \#(bd,T_2) = \#(ac,T_2)\ge 1$.
    
    \begin{figure}[tbhp]
        \centering
        \subfloat[]{
        \includegraphics[width=0.3\textwidth]{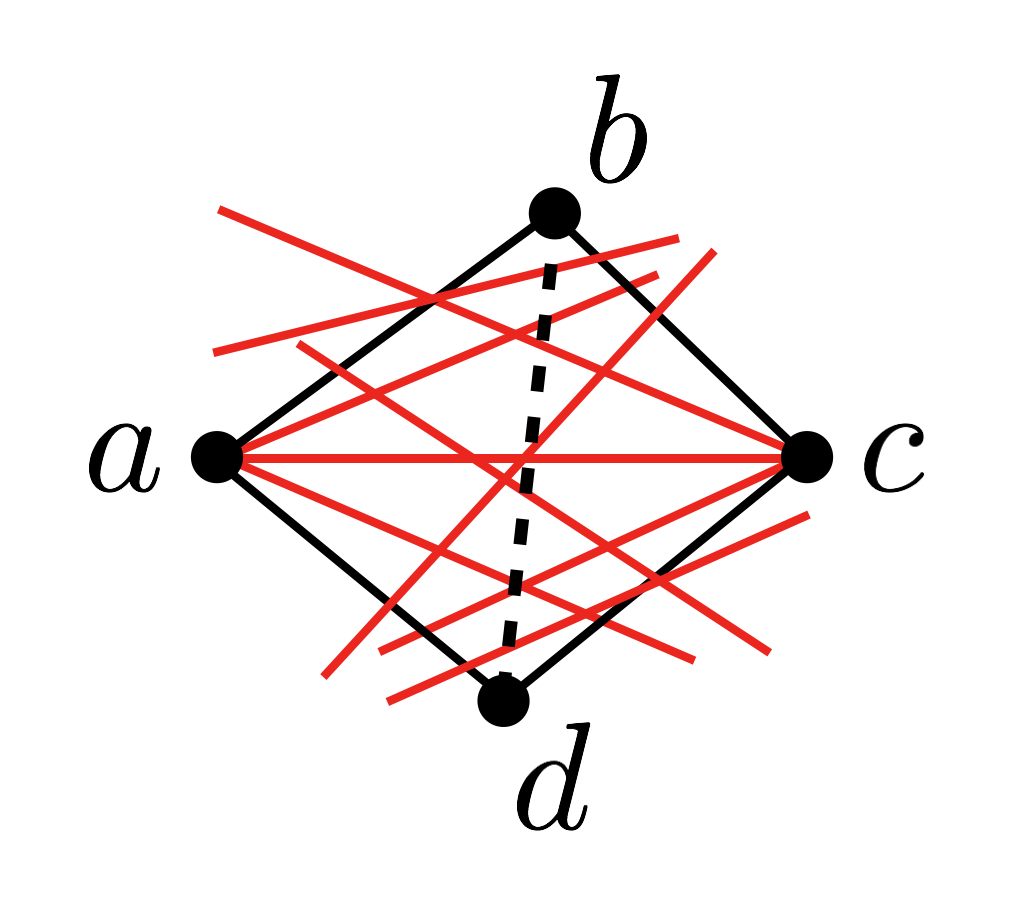}}
        \hfill
        \subfloat[]{
        \includegraphics[width=0.3\textwidth]{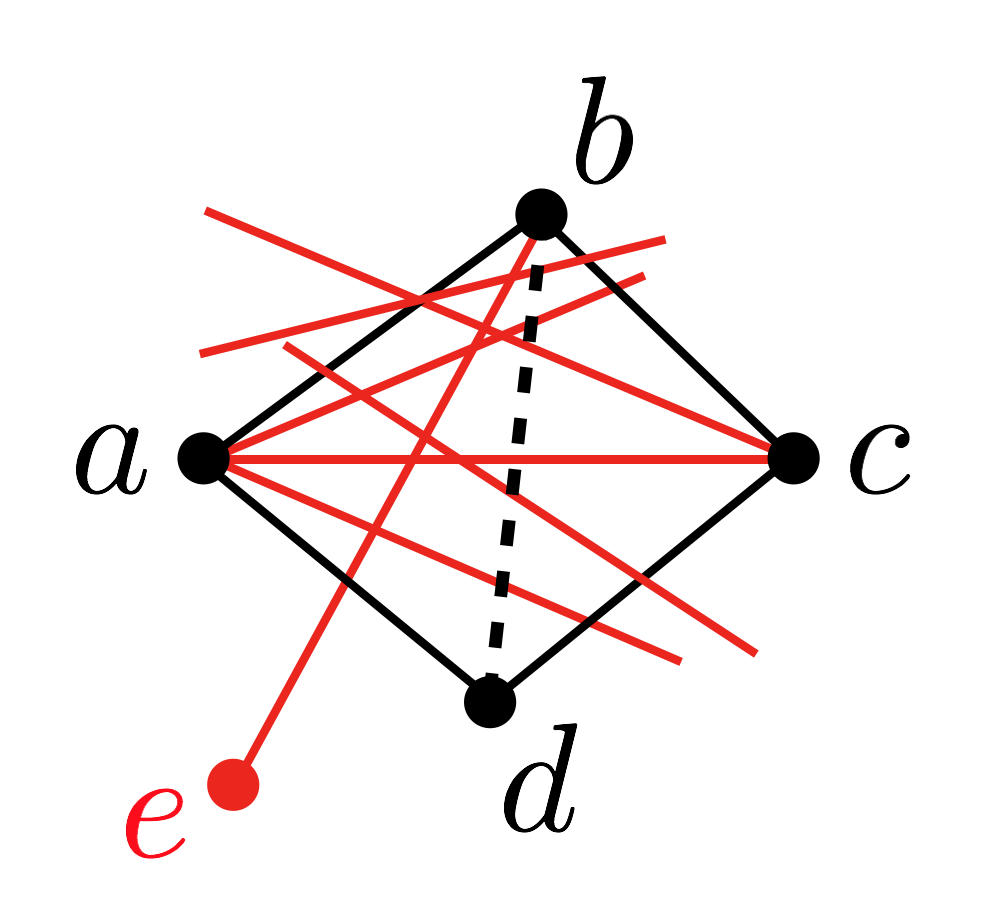}}
        \hfill
        \subfloat[]{
        \includegraphics[width=0.3\textwidth]{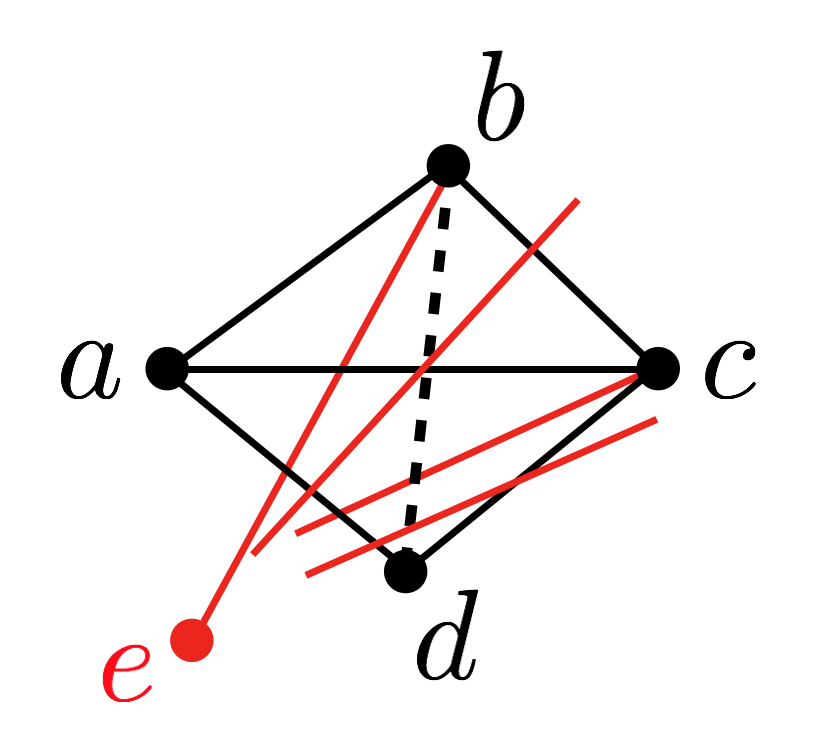}}
        \\ 
        \caption{Left: all 9 types of possible edges of $T_2$ that can intersect $bd$. Middle: the assumed existing edge $be$ along with all 6 types of possible edges of $T_2$ that can intersect $bd$ but not $ad$. All these edges must necessarily intersect $be$ inside the quadrilateral, implying that they cannot exist. Right: the assumed existing edge $be$ along with all 3 types of possible edges of $T_2$ that can intersect $bd$ and $ad$. All edges intersecting $bd$ intersect $ad$, but $be$ intersects $ad$ and not $bd$. Thus, there is at least one less intersection with $bd$ than with $ad$ and a fortiori with the maximum $ac$.}
        \label{fig: lemma 2}
    \end{figure}

    Let us now consider the other cases. The idea is to look at the edges in $T_2$ that intersect $bd$. The assumption of an existing edge coming from $b$ (or $d$) intersecting $ac$ and one other edge of the quadrilateral will imply that edges intersecting $bd$ will be restricted to intersect that same edge of the quadrilateral in order to avoid illegal intersections or meetings of edges of $T_2$ inside the quadrilateral. The concluding argument will be that $bd$ is an edge intersecting an edge of the quadrilateral but it does not intersect $bd$, which implies that $bd$ has at least one less intersection with $T_2$ than $ac$ has.
    
    Assume without loss of generality that we are in the case where $eb$ intersects $ac$ and $ad$ (see \cref{fig: lemma 2}). Let us look at edges that intersect $bd$. Note that we do not have $bd$ in $T_1$, but we can still estimate the number of intersections between this segment and $T_2$ and compare it with the number of intersections of $ac$ to justify the flip. Edges that intersect $bd$ must either intersect $ad$ or exclusive not intersect $ad$ and be one of the six following: intersect $ab$ and $bc$, pass through $a$ and intersect $bc$, pass through $c$ and intersect $ab$, intersect $ab$ and $dc$, pass through $a$ and intersect $dc$ or be $ac$. We will show that necessarily such edges must intersect $be$. First, $ac$ has at least one intersection with $T_2$, therefore it cannot be in $T_2$. Second, any edge passing through the semi-open segments $[ab[$ and $]bc]$ will intersect $eb$ inside the triangle $abc$, which is impossible since both triangulations share the same vertices. Similarly, any edge passing through the semi-open segments $[ab[$ and $]dc]$ will intersect $eb$ inside the quadrilateral, which is forbidden for the same reasons. As such, the only possibility for edges intersecting $bd$ are edges intersecting $ad$. In addition, $eb$ intersects $ad$, but it does not intersect $bd$. Therefore, $bd$ has at least one less intersections with $T_2$ than $ad$, which in turn is smaller than the maximal number of intersections: $\#(bd,T_2) \le \# (ad,T_2) -1 < \mathrm{argmax}_{\widetilde{ac}\in T_1} = \#(ac,T_2)$. Therefore, flipping $ac \xrightarrow{flip} bd$ reduces the number of intersections with $T_2$ by at least one.
\end{proof}

Note that Hanke et al. \cite{hanke1996edge} omit in their proof the case $bd$ in $T_2$ in \cref{lemma 2}. The claim and proof still hold, and it is not mathematically detrimental later in the proof as they correctly argue that ``if the edge-flipping operation $ac\xrightarrow[]{} bd$ decreases the number of intersections between the triangulations $T_1$ and $T_2$, then we are done'' and have found a maximal edge that reduces the total number of intersections, which is the goal for proving \cref{th: goal}. Indeed, $bd$ is an edge of $T_2$ provides such a case. Explicitly mentioning this case in \cref{lemma 2} greatly improves clarity of the main proof of the paper, while simultaneously yielding a sanity check that we did not forget any possible case in the case based proof. Although De Loera et al. \cite{de2010triangulations} do not explicitly write down \cref{lemma 2} (surely due to length constraints of their book), they wisely explicitly mention the case when $bd$ is an edge of $T_2$.

Before diving into the proof of the main theorem, we shall next show that if a quadrilateral has a maximal intersection diagonal, then it cannot belong to a specific class of configurations. This will later help us to find a maximal edge that  reduces the number of intersections when flipped.

\begin{lemma}
\label{lemma 2.2}
    If $abcd$ is a convex quadrilateral in $T_1$ with diagonal $ac$ of maximal intersections with $T_2$, then $T_2$ cannot have an edge $ae$ intersecting $bc$ or $cd$. Similarly, $T_2$ cannot have an edge $cf$ intersecting $ab$ or $ad$. Also, $T_2$ cannot have the edge $ac$. In summary, edges of $T_2$ that come from $a$ or $c$ cannot intersect the polygon $abcd$.
\end{lemma}

\begin{proof}
    The discussion for $ac$ is trivial. Indeed, this edge cannot be an edge of $T_2$ as we assumed $T_1\neq T_2$ (see \cref{prop: inter edge not edge other triangu}).
    
    \begin{figure}[tbhp]
        \centering
        \subfloat[]{
        \includegraphics[width=0.3\textwidth]{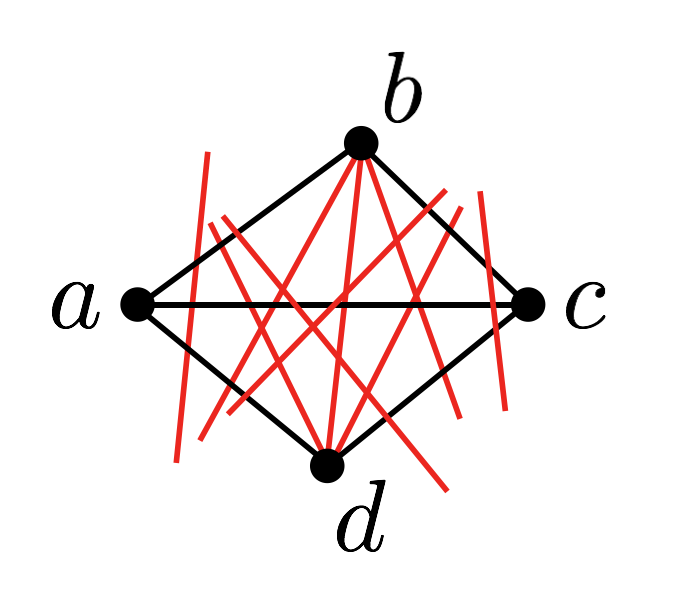}}
        \hfill
        \subfloat[]{
        \includegraphics[width=0.3\textwidth]{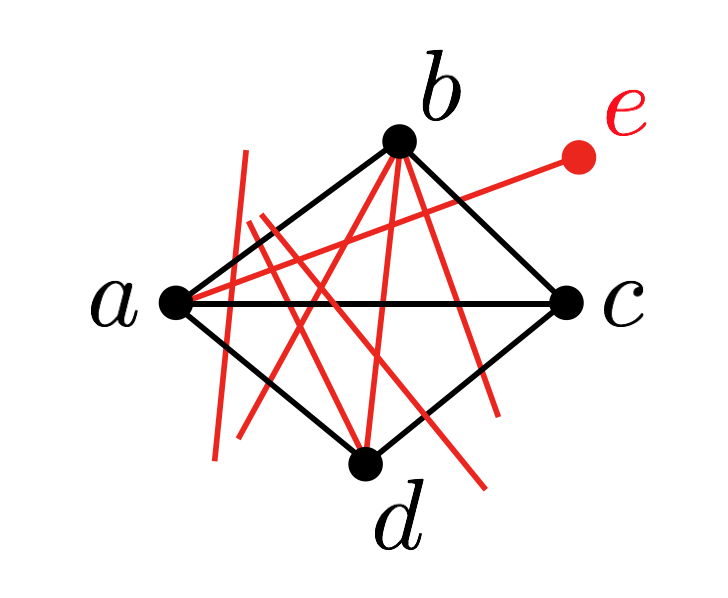}}
        \hfill
        \subfloat[]{
        \includegraphics[width=0.3\textwidth]{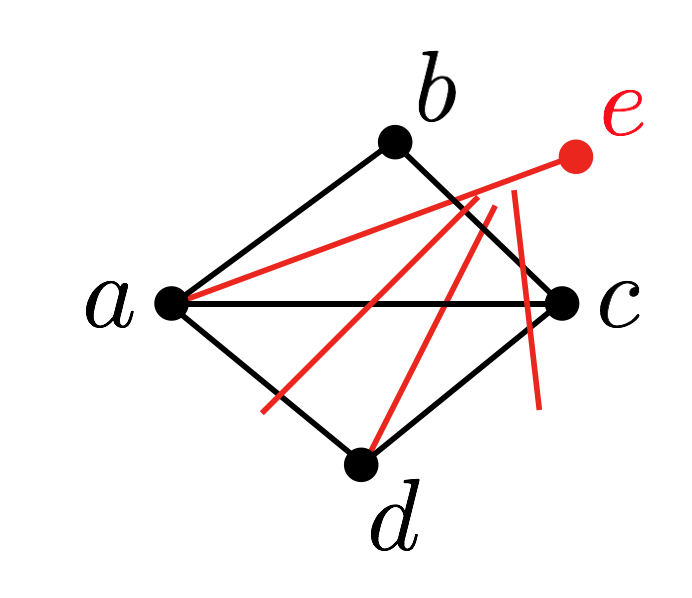}}
        \\ 
        \caption{Left: all 9 types of possible edges of $T_2$ that can intersect $ac$. Middle: the assumed existing edge $ae$ along with  all 6 types of possible edges of $T_2$ that can intersect $ac$ but not $bc$. All these edges must necessarily intersect $ae$ inside the quadrilateral, implying that they cannot exist. Right: the assumed existing edge $ae$ along with  all 3 types of possible edges of $T_2$ that can intersect $ac$ and $bc$. All edges intersecting $ac$ intersect $bc$. By maximality of $ac$, $bc$ is also maximal and no edge can intersect it that does not also intersect $ac$. However, $ae$ intersects $bc$ and not $ac$, which leads to a contradiction.}
        \label{fig: lemma 2.2}
    \end{figure}
    
    Let us now consider the other cases. Assume without loss of generality that there is an edge $ae$ in $T_2$ intersecting $bc$ (see \cref{fig: lemma 2.2}). We want to show that all edges intersecting $ac$ will intersect $bc$. This fact would imply that $bc$ would have maximal intersections with $T_2$ by maximality of $ac$ and then that $ac$ and $bc$ would intersect the same edges. This result would lead to a contradiction since $ae$ intersects $bc$ but not $ac$. 
    
    We look at the different possibilities for an edge intersecting $ac$ that does not intersect $bc$: the edge intersects $ab$ and $ad$, or it comes from $b$ and intersects $ad$, or it intersects $ab$ and comes from $d$, or it intersects $ab$ and $dc$, or it comes from $b$ and intersects $cd$, or it is simply $bd$. In any of these cases, the presence of such an edge will create an intersection with $ae$ inside the triangle $abc$, which is forbidden since both $T_1$ and $T_2$ are triangulations of the same set of points. Therefore, all edges intersecting $ac$ must intersect $bc$. Furthermore, $ae$ intersects $bc$ but does not intersect $ac$. Hence, $bc$ has at least one more intersection with $T_2$ than $ac$. This result violates the maximality assumption on $ac$. Thus, $T_2$ cannot have such an edge $ae$ intersecting $bc$.
    
    In summary, we have proven that $T_2$ cannot have an edge coming from $a$ intersecting the polygon $abcd$. Likewise, $T_2$ cannot either have an edge coming from $c$ intersecting $abcd$.
\end{proof}

Both Hanke et al. \cite{hanke1996edge} and De Loera et al. \cite{de2010triangulations} prove and use \cref{lemma 2.2} without having written it out explicitly. We believe that this choice is detrimental to clarity as this result is invoked with the same importance as \cref{lemma 2}. Furthermore, writing it out in a separate lemma avoids some awkward or unnecessarily complex formulations. Indeed, De Loera et al. \cite{de2010triangulations} claim that ``if $T_2$ has an edge which crosses the sides of the quadrilateral $abcd$ (present in $T_1$) and of which at least one of its two vertices is $a$, $b$, $c$, or $d$, then flipping $ac$ in $T_1$ decreases the number of intersections''. While the claim is mathematically accurate, the cases with such an edge with vertex $a$ or $c$ simply do not exist, which misleads the readers into thinking that they might, and should they exist, flipping the maximal edge $ac$ would reduce the number of intersections.

We are now ready for the final part of the proof. Recall that so far, we have proven that all maximally intersecting edges of $T_1$ with $T_2$ are located within a convex quadrilateral (\cref{lemma 1}), that there exists a class of configurations of these quadrilateral that guarantees that flipping that maximal edge will reduce the number of intersections (\cref{lemma 2}), and that some configurations of these quadrilaterals are forbidden (\cref{lemma 2.2}). We now have to prove that if there are no quadrilaterals with maximally intersecting diagonals satisfying the configurations of \cref{lemma 2}, then we can still find at least one maximal edge for which flipping reduces the total number of intersections with $T_2$.

\begin{lemma}
\label{lemma 3}
    There exists a convex quadrilateral $abcd$ in $T_1$ with diagonal $ac$ with maximum number of intersections with $T_2$ for which performing the flip operation $ac \xrightarrow{flip} bd$ strictly reduces the number of intersections with $T_2$.
\end{lemma}

\begin{proof}
    We once again proceed by contradiction. Assume that flipping any maximally intersecting edges with $T_2$ does not strictly reduce the total number of intersections with $T_2$. In particular, each quadrilateral with maximally intersecting diagonal will not satisfy the assumptions of \cref{lemma 2}.
    
    With this assumption in mind, we first present the sketch of the proof. Consider a quadrilateral of $T_1$ with maximally intersecting diagonal with $T_2$. Our assumptions imply that it cannot lie on the borders of the triangulations. Afterwards, we prove that there can exist at most one kind of diagonally intersecting edge of this quadrilateral. This result allows us to count how many intersections some of the edges of the quadrilateral have. In particular, it will allow us to show that there exists a $3$-zigzag of edges of the polygon, including its diagonal edge, that have maximal intersections with $T_2$. After proving that all edges of the quadrilateral have at least one intersection with $T_2$, we then show that each corner of the quadrilateral is cut by a ``corner cutter'' edge of $T_2$. We then construct a strip delimited by two sequences of vertices $u$ and $v$, such that the strip is already triangulated in $T_1$, and such that all edges from one vertex of the sequence to one of the other sequence, i.e. edges $u_iv_j$ of $T_1$, have maximal intersections with $T_2$. As the sequence can never reach the border of the triangulation, as the set of vertices is finite, and as the edges of $T_1$ do not self-intersect, we necessarily reach a cycle, and the strip has the same topology as a ring. We thus rename the strip as the ring. Since we are working on the Euclidean plane, the ring necessarily has a reflex vertex on at least one of either the $u$ sequence or the $v$ sequence (not necessarily on both). We then analyse the structure of the ring around this reflex vertex. In particular, we show that concatenations of triangles of the strip in clockwise order form a convex quadrilateral, as long as the vertices are within an angle of $\pi$ to the reference anti-clockwise border edge of the strip located at the reflex vertex. This result allows us to prove that the corner cutters of this vertex in each quadrilateral of the strip that has this vertex must intersect this reference anti-clockwise border edge of the strip, as long as the quadrilateral is within an angular reach of $\pi$ from it. By then looking at the first vertex of the strip around the reflex vertex that is not within an angular reach of $\pi$ from the reference anti-clockwise border edge, we will prove that two corner cutters of neighbouring quadrilaterals intersect within the strip, which leads to a contradiction.
    
    The existence of ``corner cutters'', to be later rigorously defined, is essential in this proof. Unfortunately, both Hanke et al. \cite{hanke1996edge} and De Loera et al. \cite{de2010triangulations} directly claim their existence without convincing arguments. They both directly claim without proof that if a quadrilateral $abcd$ has a diagonal $ac$ with maximal intersections with $T_2$, and if that quadrilateral does not fulfil the assumptions of \cref{lemma 2}, then \cref{lemma 2.2} directly implies that there exists ``corner cutter'' edges in $T_2$ for each vertex of the quadrilateral, i.e. that there exists in $T_2$ an edge intersecting $ab$ and $bc$, one intersecting $cb$ and $cd$, one intersecting $dc$ and $da$, and another one intersecting $ad$ and $ab$. While it is true that \cref{lemma 2.2} is responsible to this result, this claim is not trivial and deserves to be proven. Moreover, the proof requires an advanced familiarity with the geometry of triangulations. In order to be understandable for all, we prefer to use another approach. Furthermore, we provide a detailed proof for each of our claims.

    \begin{figure}[tbhp]
        \centering
        \subfloat[][Disjoint candidate neighbour areas]{
        \includegraphics[width=0.4\textwidth]{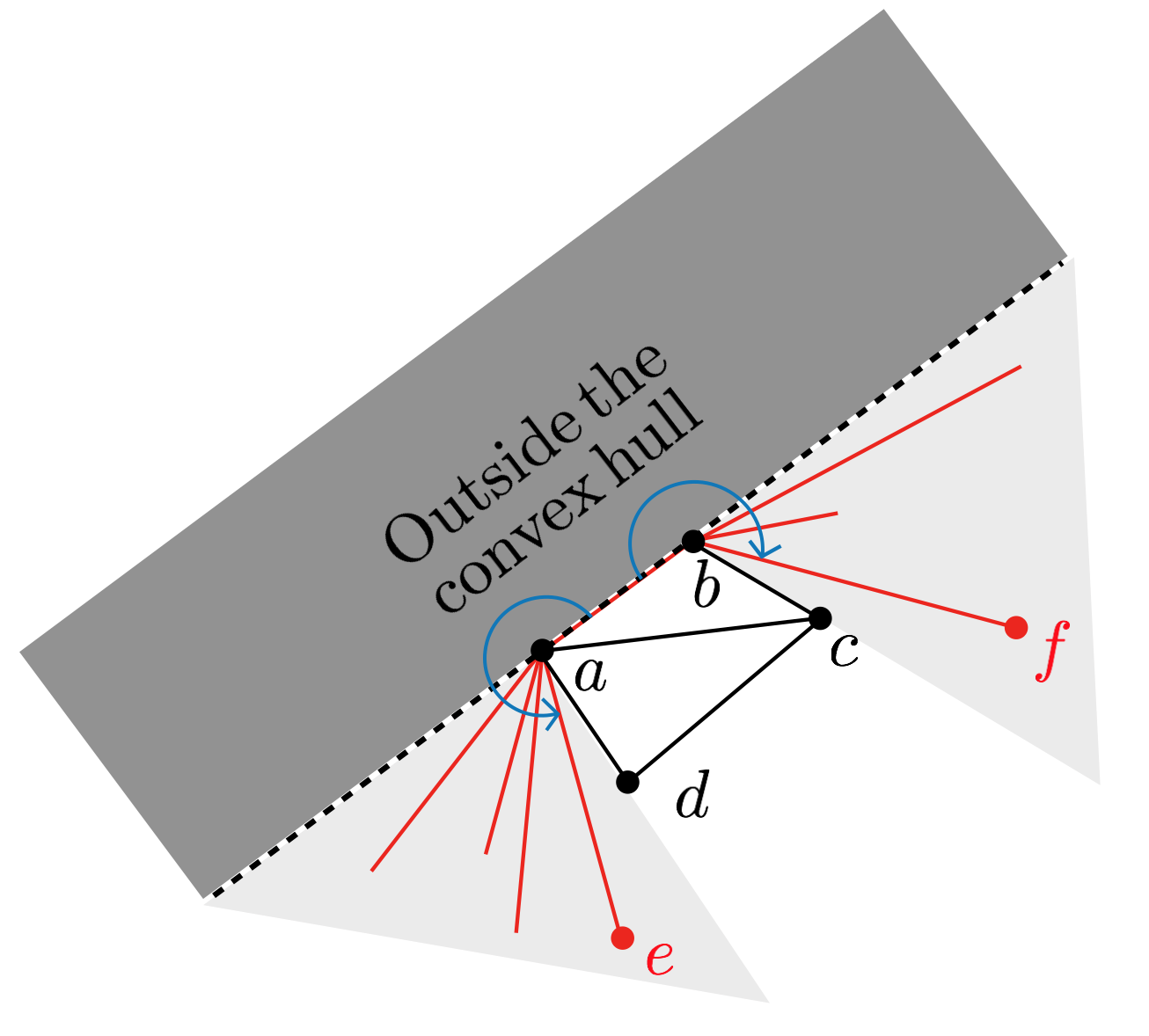}}\hfill
        \subfloat[][Intersecting candidate neighbour areas]{
        \includegraphics[width=0.4\textwidth]{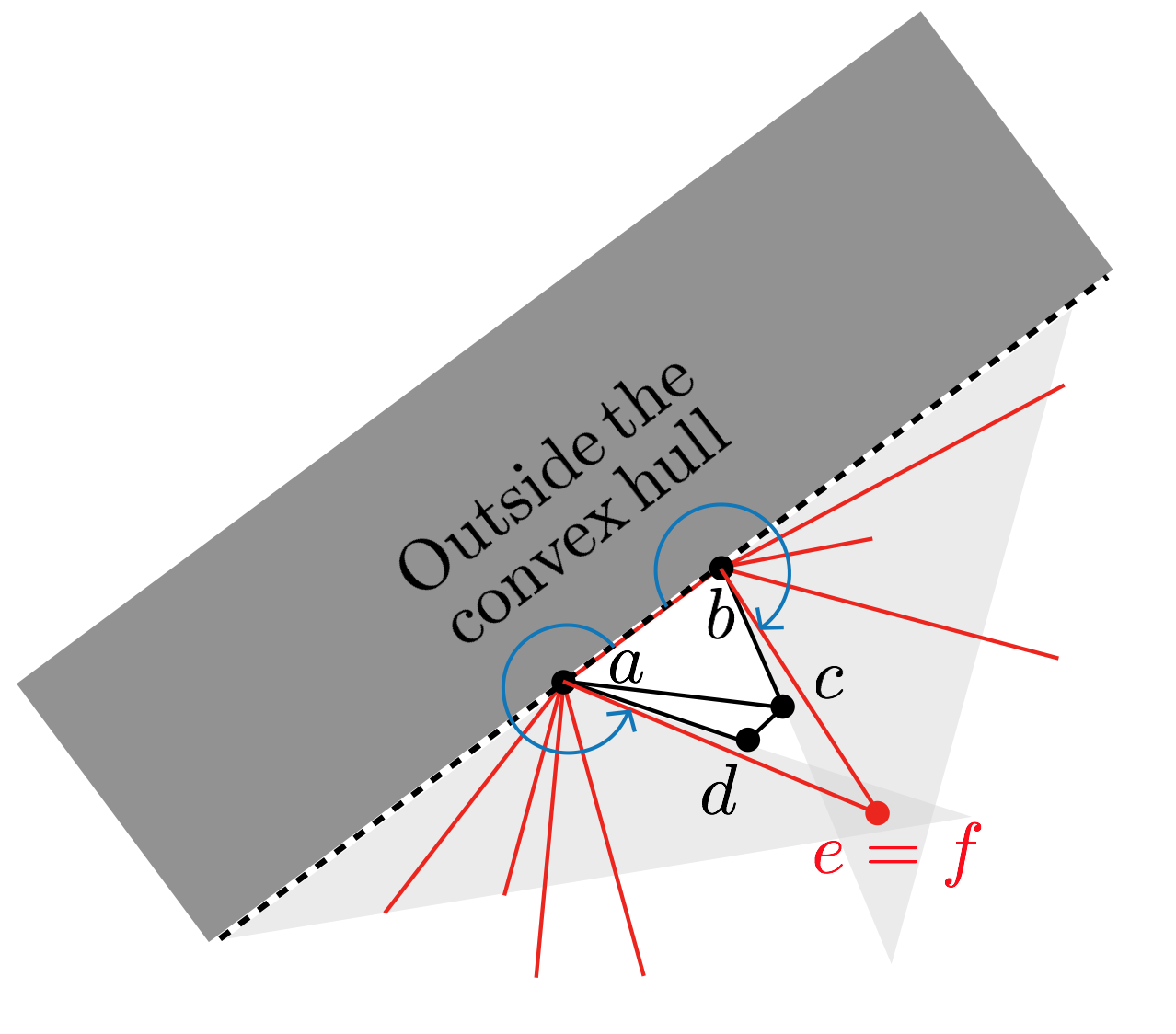}}\\ 
        \subfloat[][Disjoint candidate neighbour areas]{
        \includegraphics[width=0.4\textwidth]{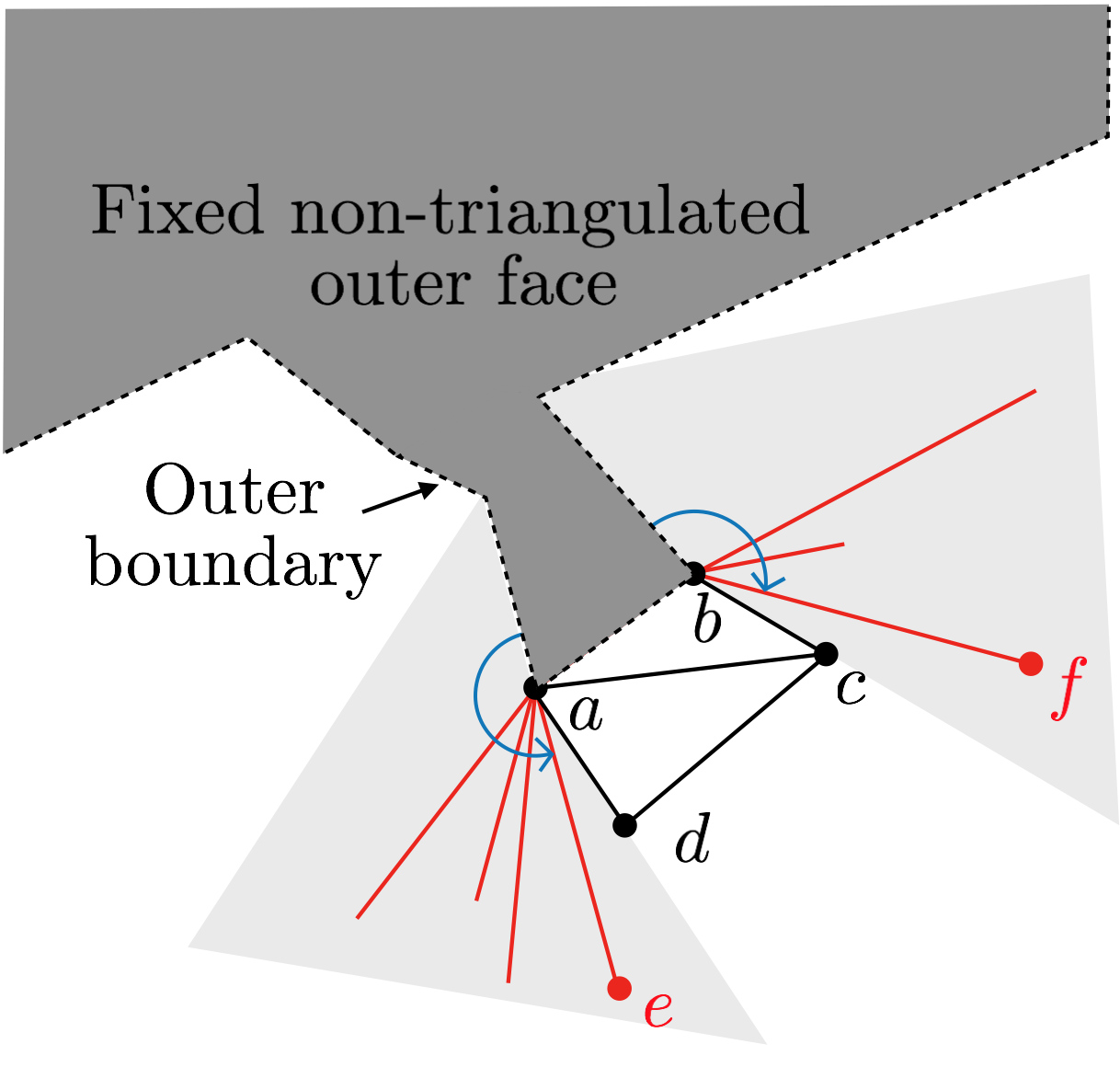}}\hfill
        \subfloat[][Intersecting candidate neighbour areas]{
        \includegraphics[width=0.4\textwidth]{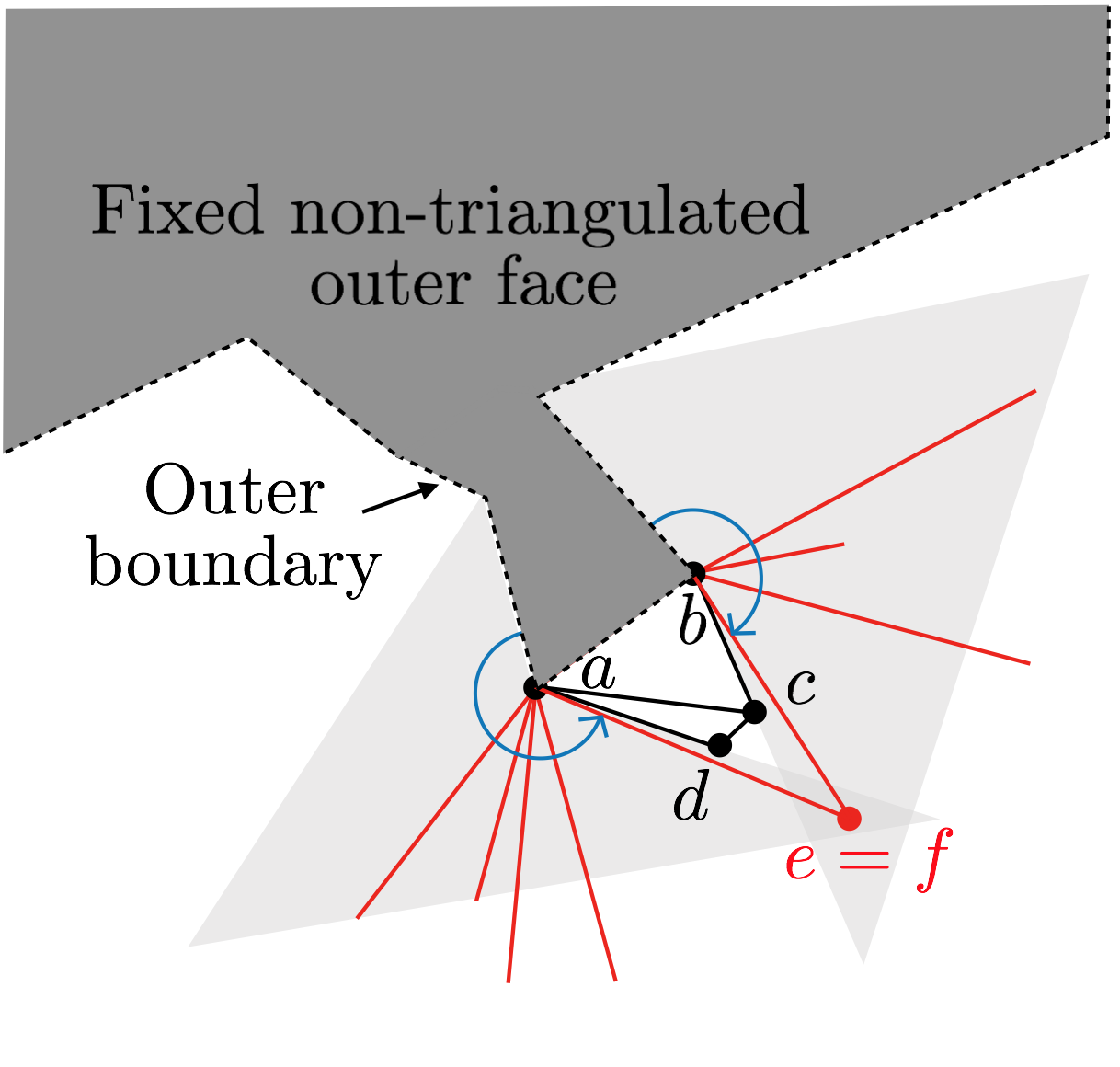}}\\ 
        \subfloat[][Disjoint candidate neighbour areas]{
        \includegraphics[width=0.4\textwidth]{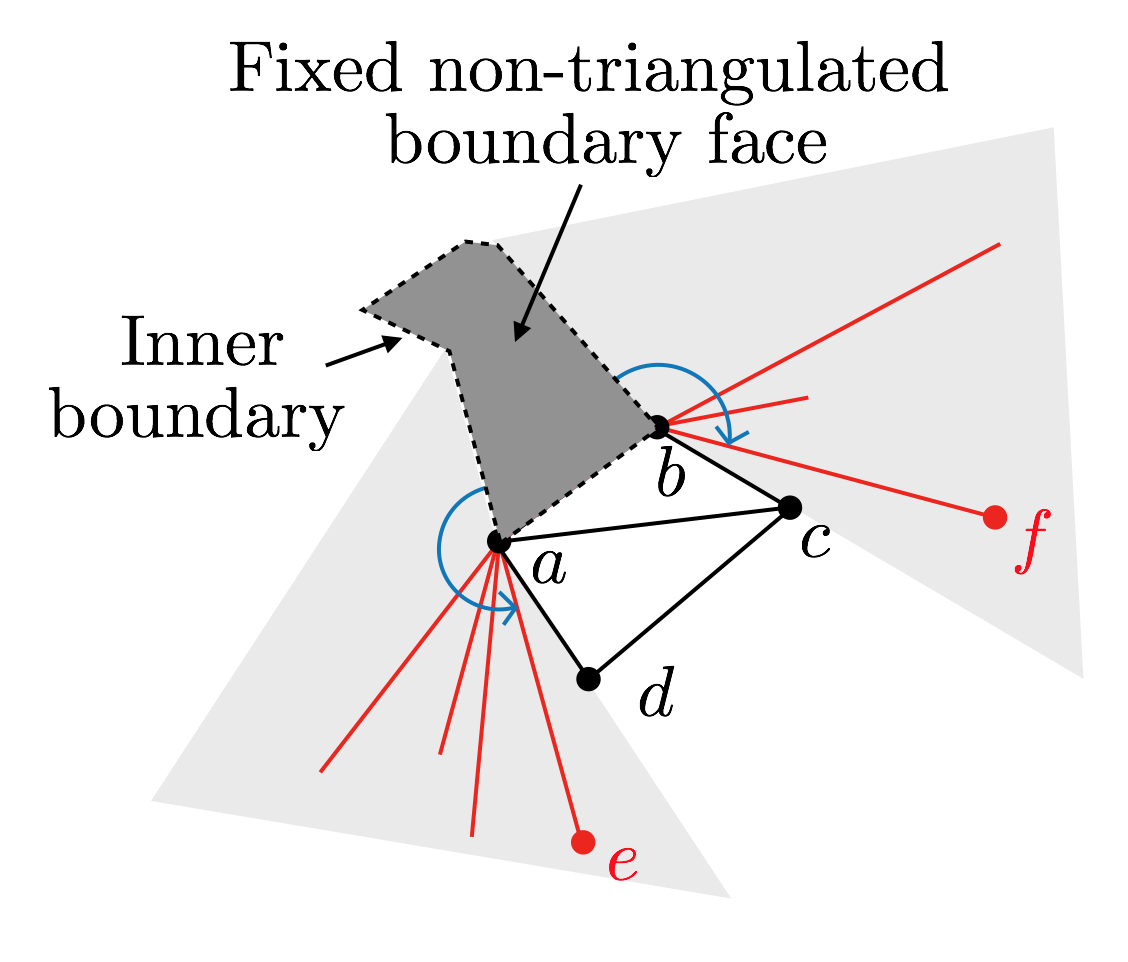}}\hfill
        \subfloat[][Intersecting candidate neighbour areas]{
        \includegraphics[width=0.4\textwidth]{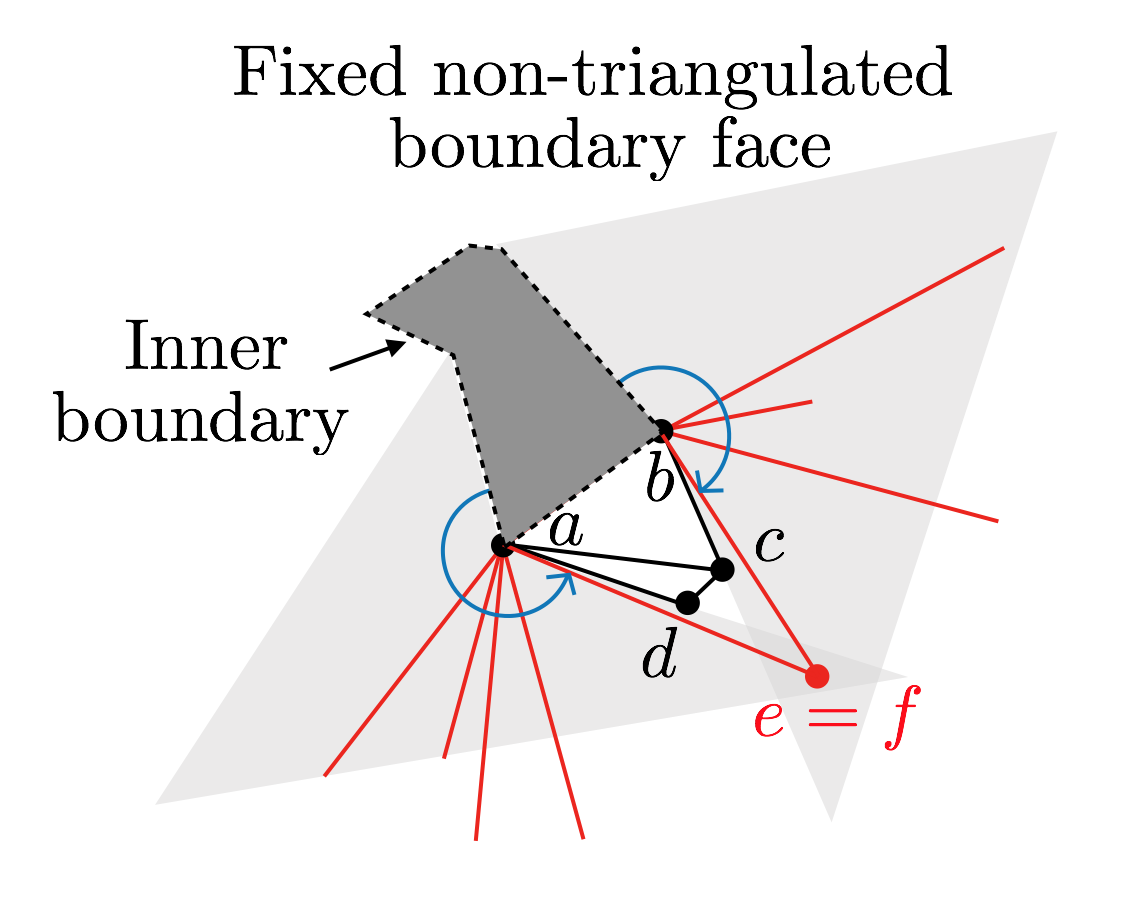}}\\ 
        \caption{Top: the edge $ab$ is on the outer border $B_0$ and on the convex hull. Middle: the edge $ab$ is on the outer border $B_0$ but is not on the polygon defined by the border of the convex hull. Bottom: the edge $ab$ is on an inner border $B_i$ defined by a fixed hole, i.e. a fixed inner non-triangulated face. Because we assume that flipping the diagonal $ac\xrightarrow{flip}bd$ does not reduce the number of intersections with $T_2$, \cref{lemma 2,lemma 2.2} constrain the domain of existence of neighbours in $T_2$ of $a$ and of neighbours of $b$. Consider $e$ and $f$ the angular extremal neighbours in $T_2$ of $a$ and $b$ closest to the quadrilateral $abcd$. Since $ab$ is on the border of $T_2$, it is an edge of $T_2$ as well and cannot have intersections with this triangulation. Thus, to maintain a triangulation, necessarily $e$ and $f$ are shared neighbours of both $a$ and $b$. This result implies that $e=f$ and in particular that the candidate areas of neighbours of $a$ and $b$ in $T_2$ intersect. We then deduce that $aeb$ is a face of $T_2$ that includes two other vertices $c$ and $d$, leading to a contradiction.}
        \label{fig: lemma 3 not border}
    \end{figure}

    We will now present the detailed proof. Let $abcd$ be a convex quadrilateral with diagonal having maximum intersections with $T_2$. We will first prove that $abcd$ cannot lie on the border of the triangulation. Note that both Hanke et al. \cite{hanke1996edge} and De Loera et al. \cite{de2010triangulations} implicitly and directly obtain this essential property by claiming the existence of ``corner cutters''.
    
    Assume $abcd$ lies on the border of the triangulation (see \cref{fig: lemma 3 not border}). Without loss of generality, assume $ab$ is a boundary edge. Then, $ab$ is necessarily an edge of $T_2$. We will then reach a contradiction by looking at the neighbours of $a$ and $b$ in $T_2$.
    
    First, since flipping $ac\xrightarrow{flip}bd$ does not reduce the number of intersections with $T_2$, we cannot have $bd$ in $T_2$, nor an edge $be$ intersecting $ac$ or $bc$ (\cref{lemma 2}).  This result means that we cannot have an edge $be$ in $T_2$, where $e\neq a$ is a vertex in the cone defined by $abc$ (in the direction of the triangle $abc$). In summary, $b$ does not have a neighbour $e$ in $T_2$ for which $be$ intersects the quadrilateral $abcd$. Likewise, $a$ has no neighbour $e$ in $T_2$ such that the edge $ae$ intersects $abcd$.
    
    Then, notice that both $a$ and $b$ have neighbours in $T_2$ different from $b$ and $a$ that lie in the half-plane delimited by $ab$ containing the quadrilateral $abcd$. Indeed, if $ab$ is on the outer boundary $B_0$ and also on the border of the convex hull, then that half-plane is where all their neighbours are. On the other hand, if $ab$ lies on the outer boundary $B_0$ but not on the border of the convex hull, or if it lies on an inner boundary $B_i$ with $i\ge 1$ (the following reasoning also applies if $ab$ is an outer boundary and on the border of the convex hull), then due to the presence of the edges $ad$ and $bc$ in $T_1$, necessarily $a$ and $b$ have neighbours in $T_2$ in that half-plane. Without loss of generality, look at $ad$. If $d$ is a neighbour of $a$ in $T_2$ then we are done. If not, then $ad$ has intersections with $T_2$. By taking the edge $ef$ with intersection point $x$ with $ad$ closest to $a$, then \cref{prop: vertex linked to the closest edge} gives that $ae$ and $af$ are edges of $T_2$. However, necessarily one of $e$ or $f$ is in the open half-plane delimited by $ab$ containing the quadrilateral $abcd$ as otherwise $ef$ would not intersect $ad$. Say $e$ is in that open half-plane. Then $f\neq b$ as $bf$ would then intersect $ad$, which is forbidden as previously proven. Therefore, $a$ has a neighbour different from $b$ in that half-plane. The same goes for $b$.
    
    We now look at the vertices in the half-open plane delimited by $ab$ that are neighbours in $T_2$ with $a$ or $b$. Let $e$ and $f$ be such neighbours of $a$ and $b$ respectively. We showed that $e$ necessarily lies in the cone $b'ab$ that does not include $abcd$, where $b'$ is a point of $\mathbb{R}^2$ that is the rotation of $b$ around $a$ of angle $\pi$ ($b'$ is on the line $ab$ but $a$ belongs to the open segment $[b'b]$). Likewise, $f$ is in the cone $a'bc$ that does not include $abcd$, where $a'\in\mathbb{R}^2$ is the rotation of $b$ around $a$ of angle $\pi$. We now fix $e$ to be such a neighbour of $a$ that is angularly closest to the edge $ad$. The orientation we are considering is the one that goes around $a$ from $ab$ to $ad$ without passing through the quadrilateral $abcd$. We can do this since we know that at least one neighbour $e$ lies in the half-plane delimited by $ab$ containing $abcd$, since there is no edge $ae$ intersecting the quadrilateral $abcd$, and since we have a finite set of points. Note that the case $e=d$ is not excluded. Similarly, we fix $f$ to be such a neighbour of $b$ closest angularly to the edge $bc$. We here considered the opposite orientation around $b$, going from $ba$ to $bc$ without passing through $abcd$.
    
    The edge $ab$ is in $T_2$. Locally, the part of the plane close to the edge inside the quadrilateral $abcd$ is part of a triangulated face of $T_1$. Therefore, it is also part of a triangulated face of $T_2$ as $T_1$ and $T_2$ triangulate the same domain. That face has $ab$ as an edge and is inside the half-plane delimited by $ab$ containing $abcd$. By closure of this face, it also contains the edge $ae$, defined as the edge as close as possible angularly to $ab$ in the orientation around $a$ going from $ab$ to $ad$ without going through $abcd$. Likewise, the face has $bf$ as an edge. Thus $a$, $b$, $e$, and $f$ are vertices of this face. Since that face has to be a triangle, which can only have $3$ vertices, we have $e=f$.
    
    Two possibilities now arise. The first is that the previously used cones $b'ad$ and $a'bc$ do not intersect in the half-plane delimited by $ab$ containing $abcd$. This leads to a contradiction as it would impose that $e$ cannot be equal to $f$. The second possibility is that these cones intersect in this half-plane, and further work is necessary to reach a contradiction. Although we could have had $e=d$ or $f=c$, we cannot have $e=c$ or $f=d$. Indeed, $ac$ has intersections with $T_2$ by maximality assumption since $T_1\neq T_2$ (\cref{prop: inter edge not edge other triangu}) and so it cannot be an edge of $T_2$, but $ae$ is an edge of $T_2$. On the other hand, as mentioned previously, $bd$ is assumed to not belong to $T_2$ as then flipping $ac\xrightarrow{flip}bd$ would strictly reduce the number of intersections with $T_2$, but $bf$ is an edge of $T_2$. Therefore, $c$ and $d$ must lie within the triangle $abe$. However, $abe$ is a face of $T_2$, so it cannot have any other vertices of the triangulation inside it. We have reached a contradiction. 
    
    We have thus proven that necessarily $abcd$ cannot lie on the border of the triangulation. In fact, the contradictory assumption in the previous proof was that $ab$ was an edge of $T_2$. Indeed, we only used the assumption that $ab$ is on the border of the triangulation to get that $ab$ was an edge of both triangulations. Therefore, we actually proved a more general statement: under our assumptions, $ab$ cannot be an edge of $T_2$, i.e. it must have an intersection with $T_2$. As $ab$ was chosen without loss of generality, we proved that all the edges of the quadrilateral are not in $T_2$, i.e. they each have at least one intersection with $T_2$. For those unconvinced, we will later reprove this result when we will need it.
   
    \begin{figure}[tbhp]
        \centering
        \includegraphics[width=0.4\textwidth]{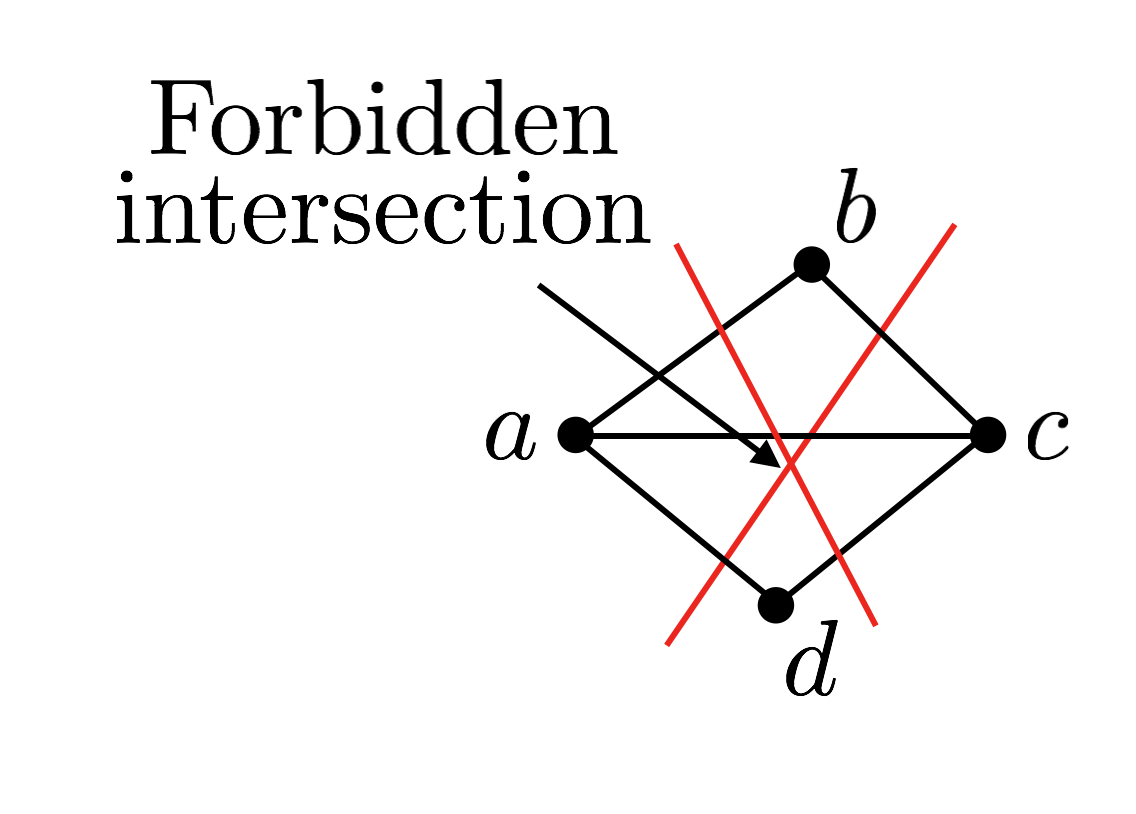}
        \caption{Not both diagonal intersecting edges can exist. We cannot have simultaneously at least one edge intersecting $ab$ and $cd$ and at least one edge intersecting $bc$ and $ad$, as otherwise there would be a forbidden intersection inside the quadrilateral. We assume from now on that there are no edges intersecting $ab$ and $cd$.}
        \label{fig: lemma 3 not both diags}
    \end{figure}
    
    \begin{figure}
        \centering
        \includegraphics[width=0.3\textwidth]{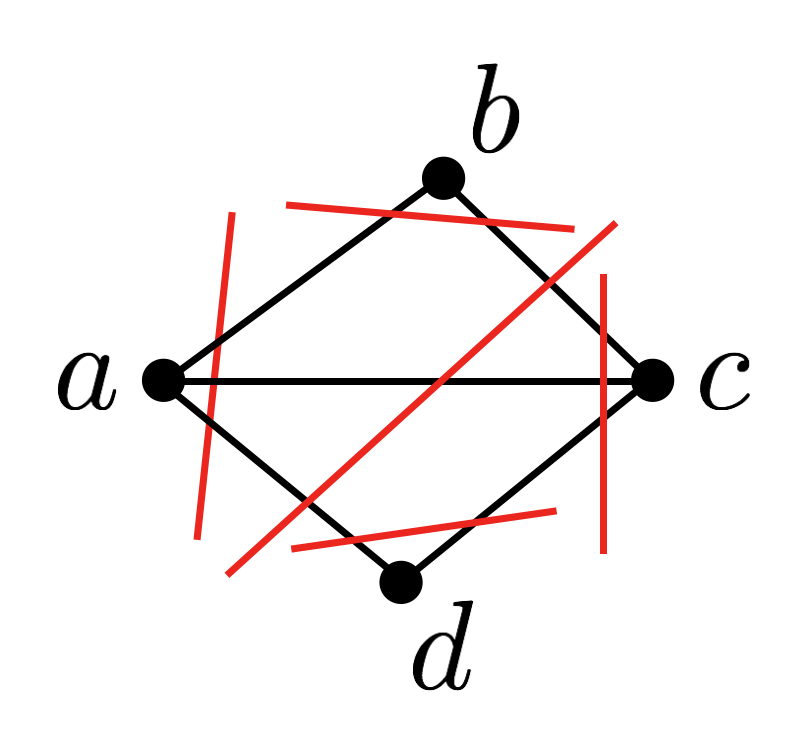}
        \caption{All five types of possible edges of $T_2$ intersecting the quadrilateral. They comprise in the corner cutters and the fixed only possible diagonal direction.}
        \label{fig: lemma 3 possibilities}
    \end{figure}

    Before continuing any further, it is important to realise that there cannot be edges of $T_2$ intersecting with the quadrilateral that pass through two different diagonals, i.e. we cannot have an edge that intersects $ab$ and $dc$ while also having another one intersecting $bc$ and $ad$ (see  \cref{fig: lemma 3 not both diags}). This would indeed imply that there is an intersection of edges of $T_2$ inside the quadrilateral, which is something forbidden. Therefore, there is at most one diagonal direction of intersection possible. Without loss of generality, we will assume from now on that the diagonally intersecting edges with the quadrilateral can only be those that intersect $ad$ and $bc$. This assumption means that we impose from now on without loss of generality that $\#(ab,cd,T_2) = 0$. We insist that this assumption is essential and invite the reader to keep it in mind. An assumption without loss of generality only claims that the proof for either case will be the same up to renaming the vertices accordingly. However, this renaming implies that all consequences derived using this general assumption must be renamed when the assumption changes. In particular, De Loera et al. \cite{de2010triangulations} forget about this important assumption in their proof which later leads them to a false claim and logic flaw.
    
    We now want to show that the edges $ad$, $ac$, and $bc$ have all maximal number of intersections with $T_2$. We will say that these edges form a $3$-zigzag of maximally intersecting edges. To prove this result, let us count the number of intersections of the edges of the $3$-zigzag and also of $bd$. 
    
    Recall that since $abcd$ does not satisfy the assumptions of \cref{lemma 2}, $bd$ cannot be an edge of $T_2$. We also assumed that $\#(ab,dc,T_2)=0$ due to our assumption that diagonally intersecting edges with the quadrilateral can only be those that intersect $ad$ and $bc$. Since $ac$ has an intersection with $T_2$, $ac$ is not an edge of $T_2$ (\cref{prop: inter edge not edge other triangu}). Due to \cref{lemma 2.2}, we also have that $\#_a(bc,T_2) = \#_a(cd,T_2) = \#_c(ad,T_2) = \#_c(ab,T_2) = 0$. By assumption, we do not satisfy the hypotheses of \cref{lemma 2}, thus $\#_b(ad,T_2) = \#_b(dc,T_2) = \#_d(ab,T_2) = \#_d(bc,T_2) = 0$. The only possibilities of edges intersecting the quadrilateral $abcd$ are displayed in  \cref{fig: lemma 3 possibilities}. We can now count intersections with $T_2$:
    \begin{alignat}{8}
        &\# (ac,T_2&) &= &\#&(ab,da,T_2&) &+ &\#&(bc,cd,T_2&) &+ &\#&(da,bc,T_2&), \label{eqn 1}\\
        &\# (bc,T_2&) &= &\#&(ab,bc,T_2&) &+ &\#&(bc,cd,T_2&) &+ &\#&(da,bc,T_2&), \label{eqn 2}\\
        &\# (da,T_2&) &= &\#&(ab,da,T_2&) &+ &\#&(da,cd,T_2&) &+ &\#&(da,bc,T_2&), \label{eqn 3}\\
        &\# (bd,T_2&) &= &\#&(ab,bc,T_2&) &+ &\#&(da,cd,T_2&) &+ &\#&(da,bc,T_2&). \label{eqn 4}
    \end{alignat}
    
    By the maximality of $ac$, we have the following implications: 
    \begin{alignat}{3}
        \cref{eqn 1} &\ge \cref{eqn 2} \implies \#(ab,da,T_2&) &\ge \#(ab,bc,T_2&), \label{eqn 5}\\
        \cref{eqn 1} &\ge \cref{eqn 3} \implies \#(bc,cd,T_2&) &\ge \#(da,cd,T_2&). \label{eqn 6}
    \end{alignat}

    By assumption, the flip operation does not reduce the number of intersections, thus $bd$ has at least as many intersections with $T_2$ than $ac$ has:
    \begin{equation}
    \label{eqn 7}
        \cref{eqn 4} \ge \cref{eqn 1} \implies  \#(ab,bc,T_2)+\#(da,cd,T_2)\ge \#(ab,da,T_2)+\#(bc,cd,T_2).
    \end{equation}
    
    By combining the previously obtained results, we have proven equality between the following quantities:
    \begin{equation}
    \label{eqn 8}
        \begin{cases}
            \cref{eqn 7} \\ \cref{eqn 5} + \cref{eqn 6}
        \end{cases}
        \hspace{-1.3em}\implies  \#(ab,da,T_2)+\#(bc,cd,T_2) = \#(ab,bc,T_2)+\#(da,cd,T_2).
    \end{equation}
    
    Using this equality, we can replace the value of $\#(ab,da,T_2)$ in one of our inequalities to obtain:
    \begin{equation}
    \label{eqn 9}
        \begin{cases}
            \cref{eqn 8} \\ \cref{eqn 5}
        \end{cases}
        \implies  \#(ab,bc,T_2)+\#(da,cd,T_2)-\#(bc,cd,T_2) \ge \#(ab,bc,T_2).
    \end{equation}
    
    By adding $\#(bc,cd,T_2)$ to this inequality, and removing $\#(ab,bc,T_2)$ on both sides, we get:
    \begin{equation}
    \label{eqn 10}
        \cref{eqn 9}+\#(bc,cd,T_2) \implies \#(da,cd,T_2) \ge \#(bc,cd,T_2).
    \end{equation}
    
    We have thus proven that in fact, the following quantities are equal:
    \begin{equation}
    \label{eqn 11}
        \begin{cases}
            \cref{eqn 6} \\ \cref{eqn 10}
        \end{cases}
        \implies  \#(bc,cd,T_2)=\#(da,cd,T_2).
    \end{equation}
    
    An analogous proof will give the symmetrical result:
    \begin{equation}
    \label{eqn 12}
        \#(ab,da,T_2)=\#(ab,bc,T_2).
    \end{equation}
    
    By substituting the two previous results in \cref{eqn 1},\cref{eqn 2},\cref{eqn 3}, and \cref{eqn 4}, we prove the desired result:
    \begin{equation}
         \#(ac,T_2) = \#(bc,T_2) =\#(da,T_2)=\#(bd,T_2).
    \end{equation}
    
    Therefore, the edges $ad$, $ac$, and $bc$ are all three of maximal intersections. Thus, the $3$-zigzag of edges of the quadrilateral $abcd$ containing the diagonal edge $ac$, and in the direction of possible diagonally intersecting edges (there are no edges intersecting simultaneously $ab$ and $cd$) is a $3$-zigzag of maximally intersecting edges with $T_2$.

    \begin{figure}[tbhp]
        \centering
        \subfloat[][Joint neighbour in the half-plane delimited by $ab$ not including $abcd$]{
        \includegraphics[width=0.45\textwidth]{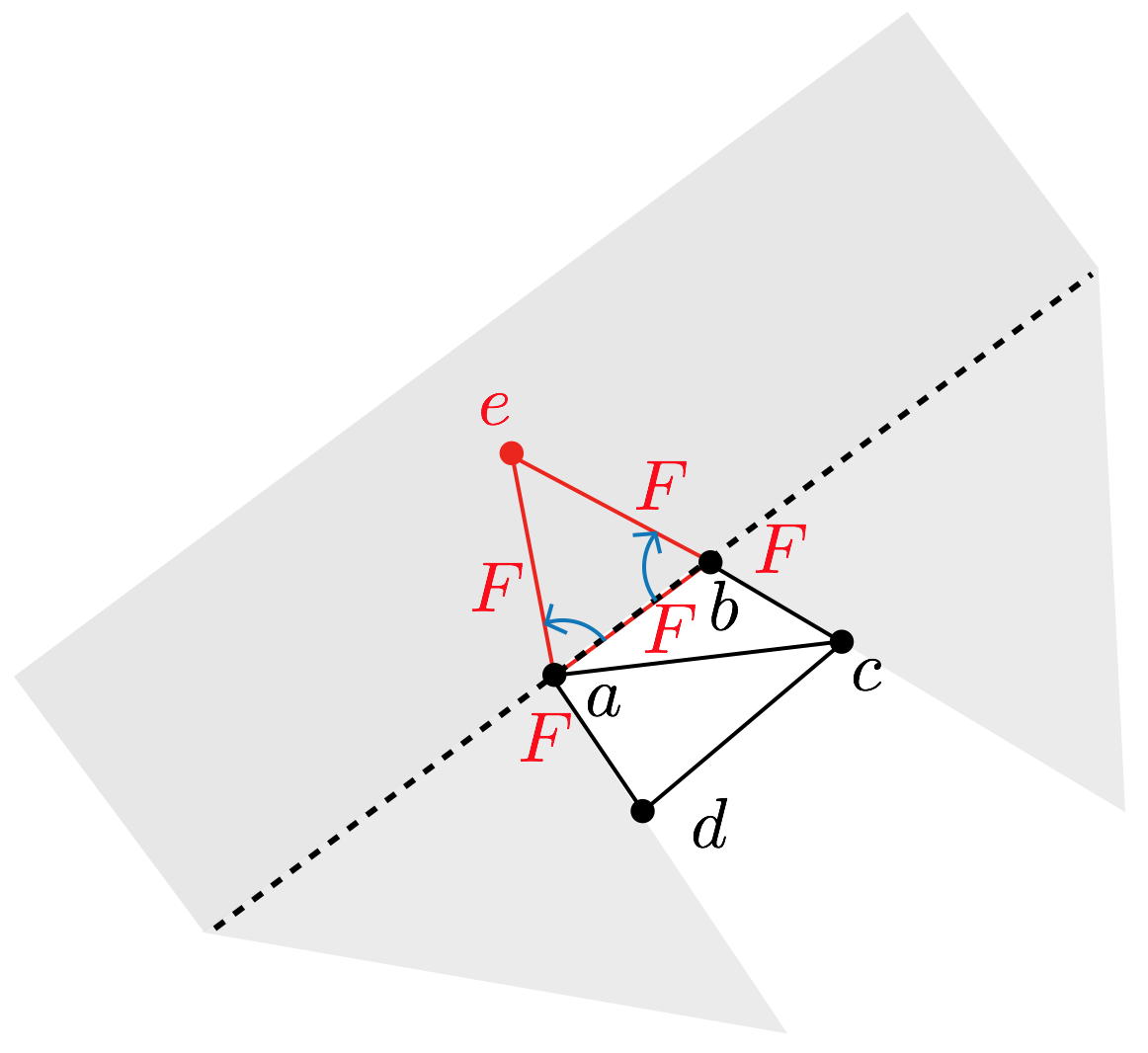}}
        \hspace{2em}
        \subfloat[][Disjoint candidate neighbour areas in the half-plane delimited by $ab$ not including $abcd$]{
        \includegraphics[width=0.45\textwidth]{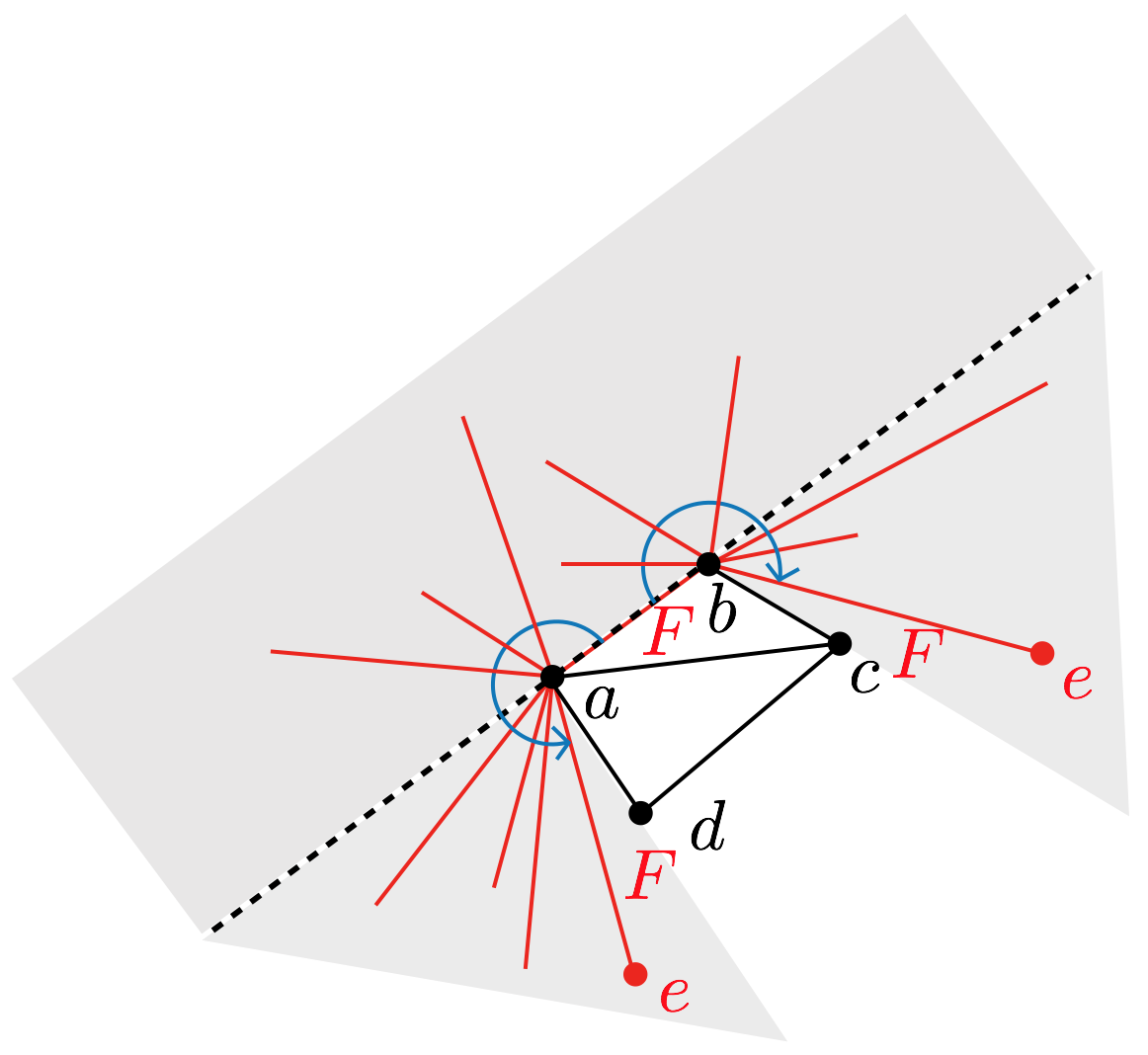}}\\
        \subfloat[][Intersecting candidate neighbour areas in the half-plane delimited by $ab$ not including $abcd$]{
        \includegraphics[width=0.45\textwidth]{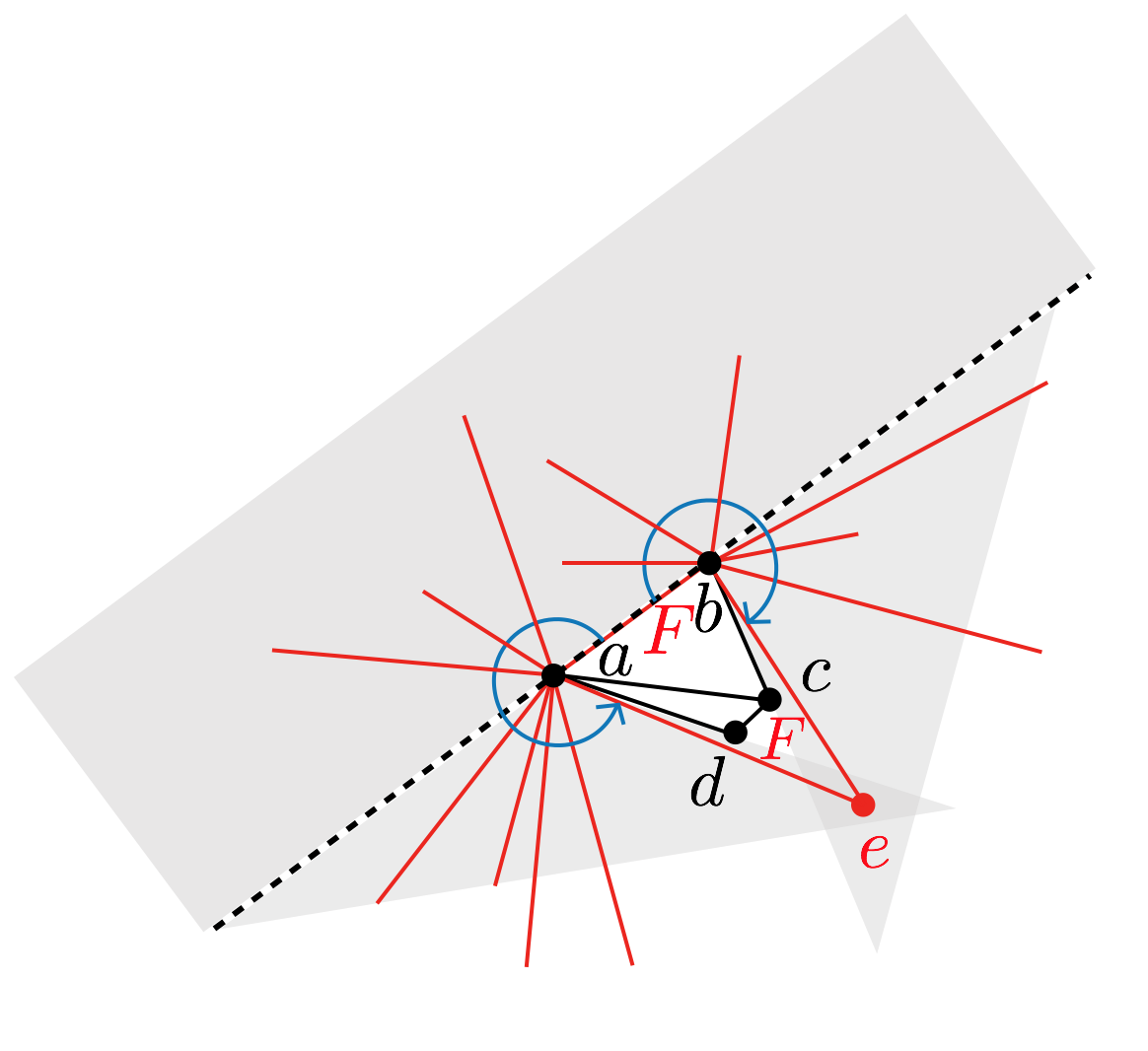}}\\ 
        \caption{Since $ab$ has no intersections with $T_2$, it is an edge of $T_2$. Denote $F$ the face of $T_2$ with edge $ab$ that is locally around $ab$ on the same half-plane delimited by $ab$ as the quadrilateral $abcd$. As $F$ is a triangle, it is then necessarily entirely included within this half-plane. In particular, its third vertex $e$, which is a common neighbour of both $a$ and $b$ in $T_2$, cannot lie on the opposite half-plane, although $a$ and $b$ might have other common neighbours in that other half-plane. As $e$ is a shared neighbour from $a$ and $b$, the domain of existence of neighbours of $a$ and $b$ must intersect. If this statement is not already contradictory, then the face $F$ contains the vertices $c$ and $d$, which leads to a contradiction.}
    \label{fig: lemma 3 all intersected}
    \end{figure}

    In the next part of the proof, we will need to know that all the edges of the quadrilateral $abcd$ are intersected. Recall that we already proved this result when proving that the quadrilateral could not lie on the border of the triangulation, as the proof there only used the border assumption to get an edge of $abcd$ that belonged to both triangulations, which then leads to a contradiction. For sceptical readers, we here re-derive a proof. Note that Hanke et al. \cite{hanke1996edge} and De Loera et al. \cite{de2010triangulations} implicitly and directly obtain this property by claiming the existence of ``corner cutters''.
    
    Without loss of generality, assume that the edge $ab$ is not intersected by $T_2$ (see \cref{fig: lemma 3 all intersected}). There is no loss of generality since we will not use the fact that we fixed the only possibility for intersecting diagonals of the quadrilateral to necessarily intersect $ad$ and $bc$.
    
    Since $T_2$ does not intersect $ab$, $ab$ is an edge of $T_2$ (\cref{prop: inter edge not edge other triangu}). As $ab$ is an edge of both triangulations, and as the area of the plane locally around $ab$ on the half-plane delimited by this edge containing the quadrilateral $abcd$ is triangulated in $T_1$, then this area is necessarily also triangulated in $T_2$ albeit differently. In particular, the area locally around $ab$ on this half-plane is part of a face $F$ of $T_2$. As $T_2$ is a triangulation, $F$ is a triangle, which implies that $F$ is entirely included in this half-plane. The vertices of $F$ are $a$, $b$, and some vertex $e$. Necessarily, $e$ is also in this half-plane.
    
    Since the quadrilateral $abcd$ does not satisfy the assumptions of \cref{lemma 2}, and since $ac$ is not in $T_2$ by maximality of $ac$ due to $T_1\neq T_2$ (\cref{prop: inter edge not edge other triangu}), the edges $ae$ and $be$ cannot intersect the quadrilateral $abcd$. Therefore, $e$ needs to be in the cone defined by $bad$ that does not include $abcd$ and it must also be in the cone defined by $abc$ that does not include $abcd$. However, $e$ must also lie on the half-plane delimited by $ab$ containing the quadrilateral $abcd$. Thus, if both previously defined cones do not intersect in this half-plane, we have reached a contradiction. If they do, then necessarily both $c$ and $d$ are included within the triangle $abe$. However, $abe$ is the face $F$ of $T_2$, which cannot have any other vertices inside it. We have reached a contradiction.
    
    We thus proved in detail that the edges $ab$, $bc$, $cd$, and $ad$ of the quadrilateral $abcd$ are not part of $T_2$, and thus each have at least one intersection with this triangulation.

    Now, we want to prove that for each vertex of the quadrilateral there is at least one ``corner cutter''. A ``corner cutter'' of $a$ for $abcd$ is defined as an edge of $T_2$ that intersects $ab$ and $ad$. The same definition is extended to other vertices of the quadrilateral.
    
    Assume that $a$ is not cut, i.e. there is no edge of $T_2$ intersecting $ab$ and $ad$: $\#(ab,da,T_2) = 0$. Using \cref{eqn 12}, we get that $\#(ab,bc,T_2)=0$.  However, the edges of $T_2$ intersecting $ab$ are those that intersect $ab$ and $ad$ and those that intersect $ab$ and $bc$ (see \cref{fig: lemma 3 possibilities}). Since non of these edges exist, $T_2$ has no edge intersecting $ab$. This statement is in contradiction with the previous result that $ab$ has at least one intersection with $T_2$. Therefore $a$ is necessarily cut. Similarly, if we assume $b$ is not cut, i.e. $\#(ab,bc,T_2)=0$, \cref{eqn 12} will give that $\#(ab,da,T_2) = 0$, which implies once again that $ab$ has no intersection with $T_2$ which is a contradiction. Thus $b$ is but. By symmetry, we have that $c$ and $d$ are also cut.
   
    We have thus proven that all the corners of the quadrilateral $abcd$ are cut. 
    
    We here take a short break in the detailed proof to remind us of what is the idea of the proof and how all the effort we have made so far is essential for the parts to come. The sketch of the proof left to reach a final contradiction is the following. We will construct a sequence of maximal edges. More formally, we will construct a triangulated ring in $T_1$, where each edge connecting a vertex on one side of the ring to a vertex on the other side of it is an edge with maximal intersections with $T_2$. Adjacent faces in this ring all have maximal edge diagonals, so we can apply the previous reasoning to them as we assumed that we could not find any such quadrilateral for which flipping the diagonal edge would reduce the number of intersections. In particular, all the quadrilaterals will have corner cutter. The reason why we have a ring is that we will sequentially create a strip-like structure that cannot reach the border of the triangulation and that cannot create self-intersections, thus by finiteness of the set of vertices we create a cycle, coined as a ring. However, since we are on the plane, a ring necessarily has one reflex vertex. By analysing corner cutters of quadrilaterals in the ring having that vertex, we will find two corner cutters that intersect inside the triangulated ring, which leads to a contradiction.
    
    We now continue the detailed proof. The structure we want to create will be defined by two sequences of connected vertices in $T_1$. We denote $(u_k)$ and $(v_l)$ these sequences that are yet to be defined. Start with $u_1 = d$, $u_2 = c$, $v_1 = a$, and $v_2 = b$. Note that the strip delimited by $u_1u_2$ and $v_1v_2$ is already triangulated in $T_1$ by the triangles $acd$ and $abc$. Further note that all edges of $T_1$ of the form $u_iv_j$ for $(i,j)\in\{1,2\}^2$ are in $T_1$, and they all have maximal intersections with $T_2$. The task is to incrementally find $u_{k+1}$ or $v_{l+1}$ given $u_1,u_2,\cdots,u_k$ and $v_1,v_2,\cdots,v_l$ that will maintain the property of the curved strip delimited by $u$ and $v$ is already triangulated in $T_1$, and for which all edges of $T_1$ of the form $u_i v_j$ have maximal intersections with $T_2$. Proving how to do the construction from $k=l=2$ to $k+1=3$ or $l+1=3$ will explain how to perform the incremental construction for general $k$ and $l$. See \cref{fig: lemma 3 strip construction} for an illustration.
    
    \begin{figure}[tbhp]
        \centering
        \includegraphics[width=0.5\textwidth]{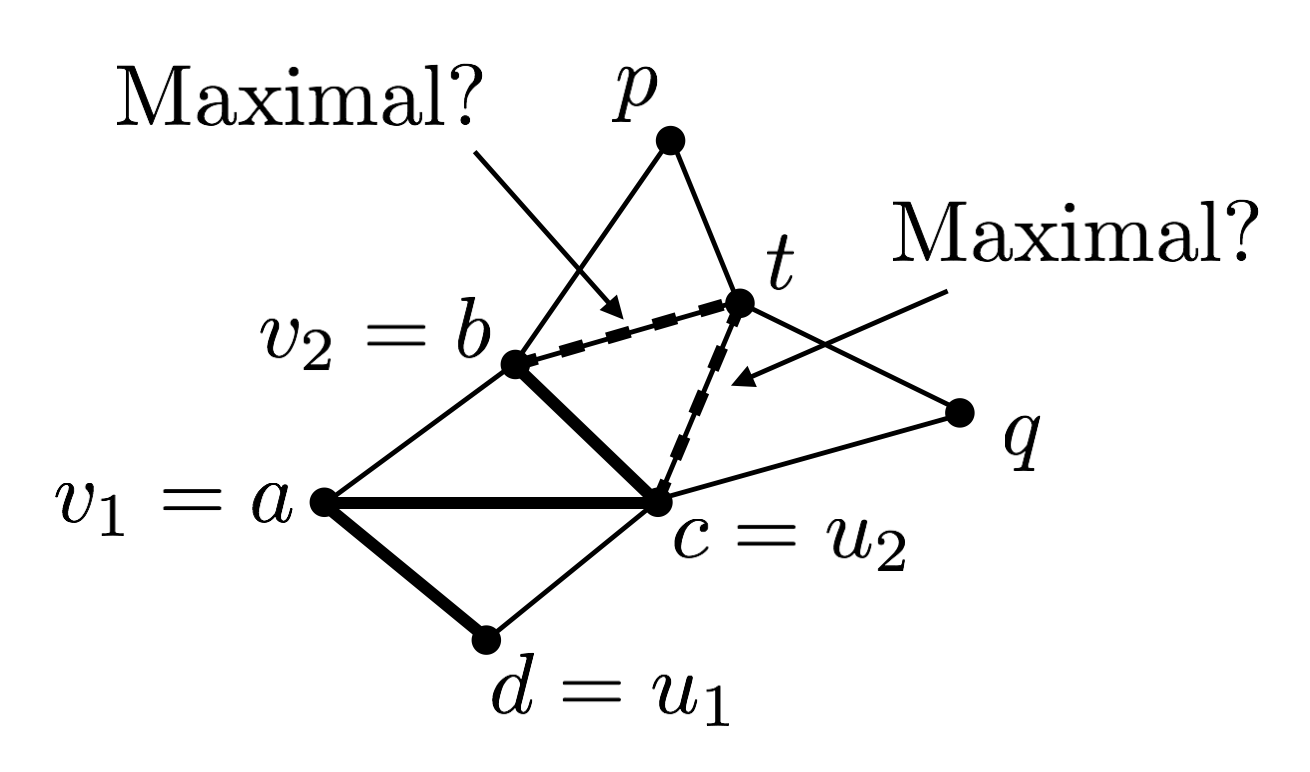}
        \caption{Since the quadrilateral $abcd$ is not on the borders of the triangulation, then $b$ and $c$ share a common neighbour $t$ different from $a$. As we showed that $bc$ has maximal intersections with $T_2$, the quadrilateral $abtc$ is a new quadrilateral with diagonal having maximal intersections with $T_2$. Furthermore, the assumption that flipping its diagonal does not result in a strict decrease in the number of intersections with $T_2$ holds for this new quadrilateral. Thus, by reapplying the same proof as we did on $abcd$, $abtc$ has a $3$-zigzag of maximally intersecting edges. In particular, $bt$ or $ct$ must also have maximal intersections with $T_2$. If $bt$ is maximally intersecting $T_2$, then add $u_3 = t$ to the $u$ sequence of the strip. If it is not, then $ct$ is maximally intersecting $T_2$, and then add $v_3$ to the $v$ sequence of the strip. We can reiterate the reasoning in the next quadrilateral, which exists as $abtc$ cannot lie on the border of the triangulations, and which has a maximally intersecting diagonal with $T_2$. This allows us to construct incrementally the desired strip.}
        \label{fig: lemma 3 strip construction}
    \end{figure}
    
    We previously showed that, because we assumed without loss of generality that there are no edges of $T_2$ intersecting simultaneously $ab$ and $cd$, the 3-zigzag of edges $ad$, $ac$, and $bc$ of the quadrilateral $abcd$ passing through the diagonal $ac$ all have maximal intersections, i.e. $ad$ and $bc$ are also maximal intersection edges. Since $abcd$ is not on the border of the triangulations, the edge $bc$ is the diagonal of a quadrilateral in $T_1$, i.e. there exists a vertex $t\neq a$ such that $bt$ and $ct$ are edges of $T_1$. This new quadrilateral $abtc$ has its diagonal $bc$ with maximal intersections with $T_2$, which implies that it is convex due to \cref{lemma 1}. Furthermore, we assumed that all such quadrilaterals could not have a diagonal for which a flip would strictly reduce the number of intersections. Therefore, we can redo the same reasoning on $abtc$ that we did for $abcd$. The induction is not as trivial as it seems and it is easy to make mistakes and omit important cases. Both Hanke et al. \cite{hanke1996edge} and De Loera et al. \cite{de2010triangulations} make major logic flaws by applying the same reasoning too hastily. By applying the induction, there exists a $3$-zigzag of edges of $abtc$ passing through its diagonal that are all with maximum intersections with $T_2$. There are two possible $3$-zigzags for $abtc$: the first one consists of the edges $ac$, $bc$, and $bt$, the second consists of $ab$, $bc$, and $ct$. Without further assumptions, we do not know which of these two $3$-zigzags are maximal, even though we assumed that the only diagonally intersecting edges of $abcd$ are those intersecting $ad$ and $bc$. We will look at the two only cases: either $bt$ has maximal intersections with $T_2$, or it does not and then $ct$ has maximal intersections with $T_2$.

    \begin{figure}[tbhp]
        \centering
        \includegraphics[width=0.6\textwidth]{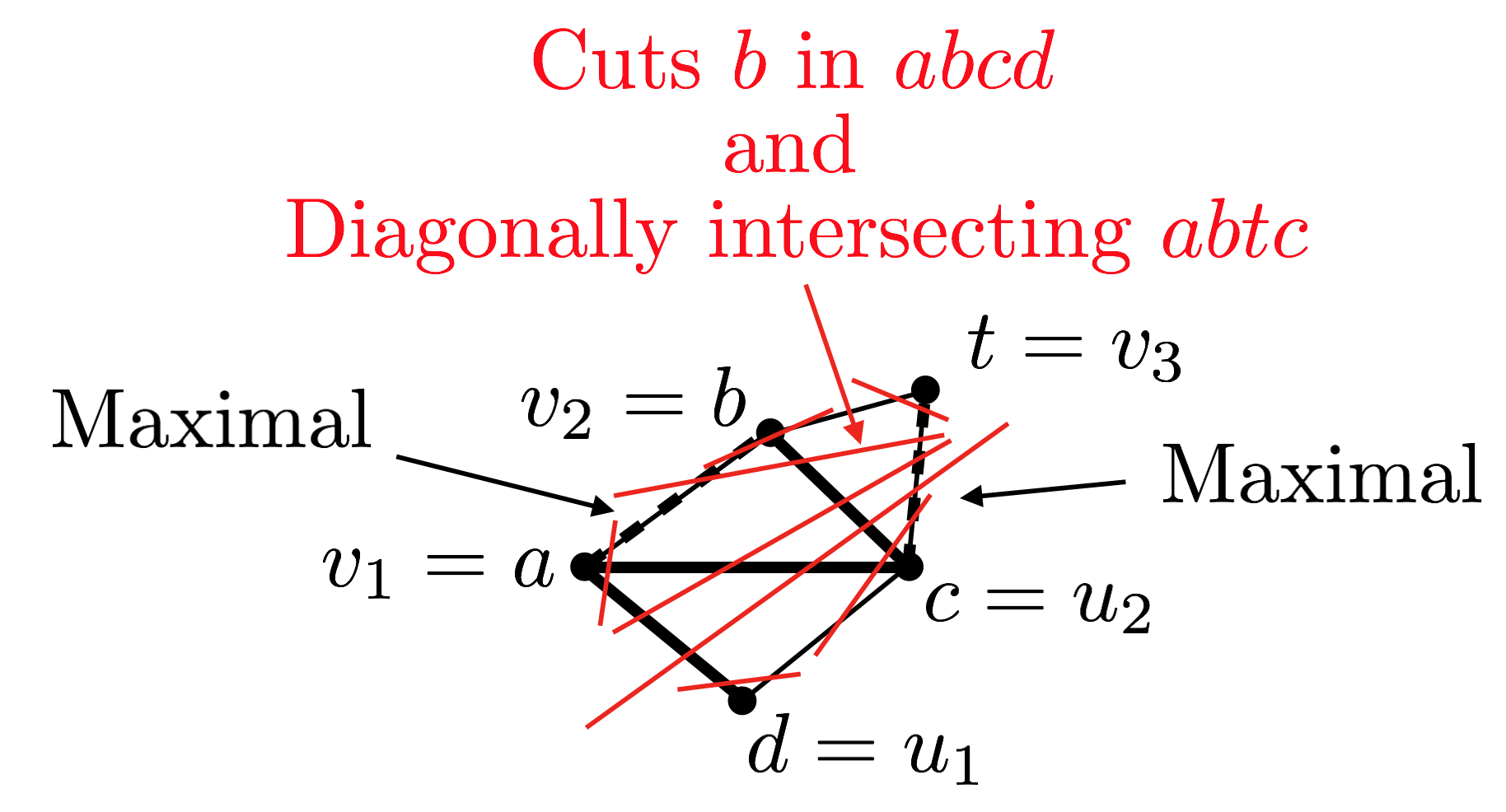}
        \caption{We cannot exclude the case where although $ad$, $ac$, and $bc$ form a 3-zigzag of maximally intersecting edges with $T_2$, since there is an edge of $T_2$ diagonally intersecting $abcd$ by intersecting $ad$ and $bc$ by assumption, that the 3-zigzag of maximally intersecting edges with $T_2$ in $abtc$ is $ab$, $bc$, and $ct$ and not $ac$, $bc$, and $bt$. It is possible that all diagonally intersecting edges of $abcd$ are corner cutters in the next quadrilateral whereas corner cutters of one vertex of $abcd$ become diagonally intersecting edges of $abtc$. This new zigzag does not necessarily create a trivial conflict with corner cutters of either $abcd$ or $abtc$. There is thus no apparent contradiction in this case.}
        \label{fig: lemma 3 second case zigzag must be careful}
    \end{figure}

    De Loera et al. \cite{de2010triangulations} here forget one of the cases, which is the most important gap in their book version of the proof. This omission is due to the fact that they overlook the fact that the 3-zigzag was obtained under the fundamental assumption that edges of $T_2$ diagonally intersecting $abcd$ could only intersect $ad$ and $bc$, i.e. the assumption that $\#(ab,cd, T_2) = 0$. Indeed, they implicitly assume that the 3-zigzag extends to $ac$, $bc$, and $ct$ in $abtc$. But the direction of the diagonal is given by the direction of possible diagonally intersecting edges of $T_2$ with the quadrilateral. Unfortunately, having that all edges of $T_2$ diagonally intersecting $abcd$ intersect $ad$ and $bc$ does not imply that they all then intersect as well $ac$ and $bt$. In fact, it can be possible that none of these edges intersect $bt$. In such cases, we could actually have that all edges of $T_2$ diagonally intersecting $abtc$ intersect $ab$ and $ct$ instead. Even by taking into account the existence of corner cutters, there are cases where this other direction of 3-zigzag does not immediately imply a contradiction. Indeed, edges diagonally intersecting $ad$ and $bc$ could all be corner cutters of $c$ in $abtc$, corner cutters of $b$ in $abcd$ could all be diagonally intersecting $abtc$ by intersecting $ab$ and $ct$, and corner cutters of $c$ in $abcd$ could all be corner cutters of $c$ in $abtc$. 
    Because of their mistake, fixing the initial $3$-zigzag incorrectly fixes for the authors inductively a sequence of zigzags with a new maximal edge with vertex the latest added vertex to the sequence. For this reason, De Loera et al. \cite{de2010triangulations} only introduce the term ``zigzag''. As we have seen, such a zigzagging only occurs between three edges of a quadrilateral and does not generalise: a new vertex from a 3-zigzag may not be a vertex of the next maximal edge in the sequence. To fix this, we introduced the concept of a $3$-zigzag rather than a zigzag. Overall, this error shows the importance of rigorous attention to details in the derivation of the proof, rather than using very quick and seemingly reasonable arguments justified by misleading figures. These turn out to be possibly restricted to some particular constellations but are false for general ones when all cases are carefully scrutinised. See \cref{fig: lemma 3 second case zigzag must be careful} for an illustration of such a possible but, so far, overlooked case.

    We return to the discussion on which of the two $3$-zigzags are maximal. Consider the first case: $bt$ has maximal intersections with $T_2$. We then choose $u_{k+1} = u_3 = t$. The strip delimited by $u_1u_2\cdots u_{k+1}$ and $v_1v_2\cdots v_l$ is already triangulated in $T_1$, and each edge of $T_1$ of the form $u_i v_j$ has maximal intersections with $T_2$. Note that we did not create yet $v_{l+1} = v_3$. As we showed previously for $abcd$, the quadrilateral $abtc$ cannot lie on the border of the triangulation, therefore $b$ and $t$ have a second shared neighbour, i.e. there exists a vertex $p\neq c$ such that $bp$ and $tp$ are edges of $T_1$. This new quadrilateral has maximally intersecting diagonal. We can thus continue from here the construction of $u$ and $k$ by induction.
    
    Consider the second case: $bt$ does not have maximal intersections with $T_2$. Since $abtc$ must have a $3$-zigzag of maximally intersecting edges including its diagonal edge, then $ct$ has maximal intersections with $T_2$. We thus choose $v_{l+1} = v_3 = t$, but do not create yet $u_{k+1}$. The strip delimited by $u_1u_2\cdots u_k$ and $v_1v_2\cdots v_{l+1}$ is already triangulated in $T_1$, and each edge of $T_1$ of the form $u_i v_j$ has maximal intersections with $T_2$. As we showed previously for $abcd$, the quadrilateral $abtc$ cannot lie on the border of the triangulation, therefore $c$ and $t$ have a second shared neighbour, i.e. there exists a vertex $q\neq b$ such that $cq$ and $tq$ are edges of $T_1$. This new quadrilateral has maximally intersecting diagonal. We can thus continue from here the construction of $u$ and $k$ by induction.
    
    Using this process, we can incrementally construct the sequences $u$ and $v$ until we reach a cycle, that is until we reach a point where the current $k,l$ points are $u_k = u_1$ and $v_l = v_1$ for $k$ or $l$ different to $1$. Indeed, reaching a cycle is unescapable. Since the quadrilaterals being considered at each step always have a maximally intersecting diagonal with $T_2$, then they can never lie on the border of the triangulation, which implies that we can always continue the construction process. However, the number of vertices, and thus of edges, is finite. Therefore, at some point, we will have cycled. Hanke et al. \cite{hanke1996edge} missed the fact that cycling is a possibility and directly claim that this inductive process necessarily reaches the border and thus leads to a contradiction. De Loera et al. \cite{de2010triangulations} correct this by analysing a cycle of maximal edges. Unfortunately, they incorrectly constructed this cycle with a global zigzag of maximal edges, but in fact, a correct analysis of the 3-zigzags as we did, shows that the cycle construction does not provide a global zigzag of maximal edges. This logic error is the greatest flaw in their proof. See \cref{fig: lemma 3 ring} for an example of ring of maximal edges, constructed by analysing local 3-zigzags of maximal edges, but for which the maximal edges do not globally zigzag throughout the ring.

    \begin{figure}[tbhp]
        \centering
        \includegraphics[width=0.5\textwidth]{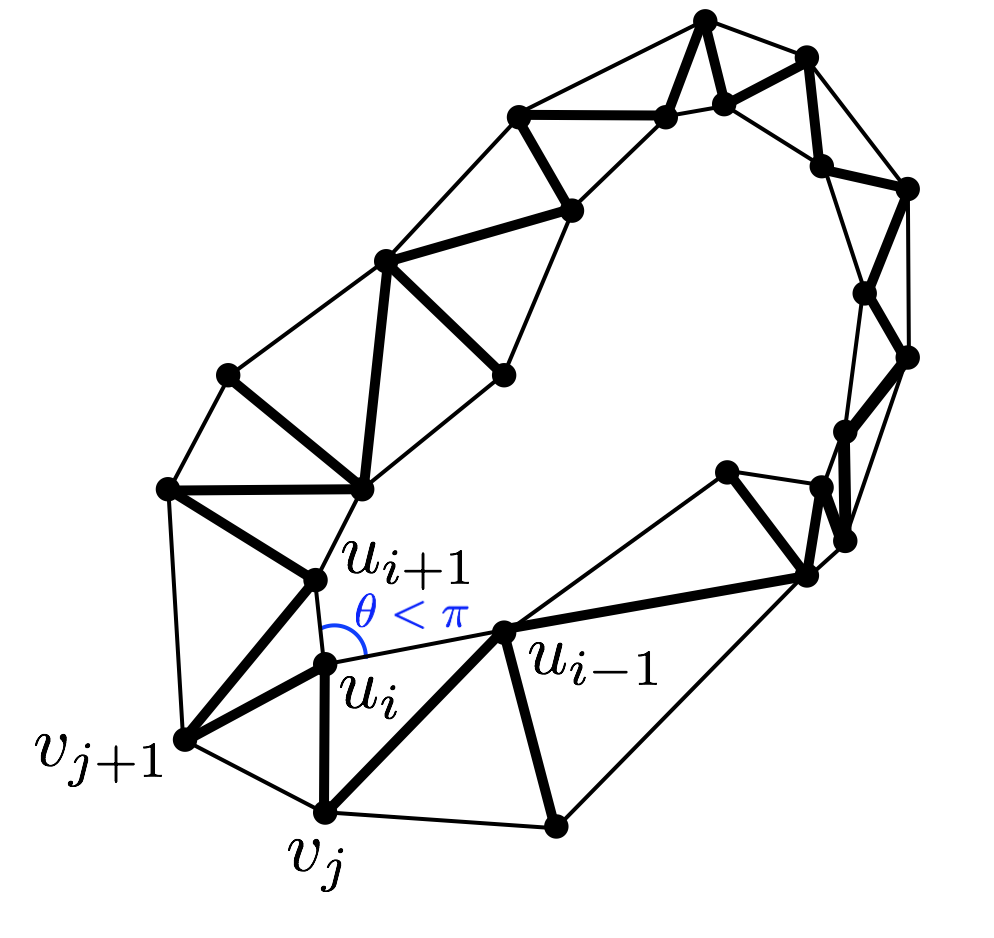}
        \caption{The strip delimited by $u$ and $v$ with edges of $T_1$ between both sequences, i.e. edges of $T_1$ of the form $u_iv_j$, all have maximal intersections with $T_2$. As the strip cannot reach the border of the triangulation, by finiteness of the problem it necessarily cycles. Furthermore, since it cannot self-intersect, and since we are in the Euclidean plane, the strip has the topology of a ring, and it necessarily has at least one reflex vertex.}
        \label{fig: lemma 3 ring}
    \end{figure}
    
    As the strip is delimited by $u_1\cdots u_ku_1$ and $v_1v_2\cdots v_lv_1$ that are edges of $T_1$, then there necessarily is one vertex $u_i$ or $v_j$ that is a reflex vertex. See \cref{fig: lemma 3 ring} for an illustration. We say $u_i$ is a reflex vertex of the strip if $\widehat{u_{i-1}u_i u_{i+1}} < \pi$, where $\widehat{u_{i-1}u_i u_{i+1}}$ is the geometric angle of the cone $u_{i-1}u_i u_{i+1}$ that does not contain the triangles of the strip with at least one vertex among $u_{i-1}$, $u_i$, or $u_{i+1}$. It can be seen as the ``outer'' angle from the strip at $u_i$ is strictly smaller than $\pi$. Similarly, we say $v_j$ is a reflex vertex of the strip if the ``outer'' angle from the strip at $v_j$ is strictly smaller than $\pi$. If the strip did not have a reflex vertex, then necessarily it would intersect one of the edges $u_i v_j$, which is forbidden as the strip is composed only of edges of the planar triangulation $T_1$. We thus have that the strip has the topology of a ring. 
    
    Without loss of generality, assume that the ring has a reflex vertex at a vertex of the $u$ sequence, denoted $u_i$. By renumbering the $v$ vertices, let $v_{-1}, v_0, v_1, \cdots v_r $ be the vertices $v_j$ in clockwise order such that $u_i v_j$ are edges of $T_1$. Note that $r\ge 0$, i.e. that there are at least $2$ vertices of the $v$ sequence connected to $u_i$ in $T_1$ on the strip. Indeed, by construction of the strip, each vertex of $u$ is connected to at least one vertex of $v$ in $T_1$, and vice versa each vertex of $v$ is connected to at least one vertex of $u$ in $T_1$. If $u_i$ is connected to only one vertex of $v$ in $T_1$, say $v_j$, then look at the quadrilateral of the strip $u_i u_{i-1} v_j u_{i+1}$, which has $u_i v_j$ as maximally intersecting diagonal with $T_2$. This quadrilateral is convex (\cref{lemma 1}). This property implies that the angle formed by the $u_{i-1}u_iu_{i+1}$ including the quadrilateral is smaller than or equal to $\pi$. Therefore, such vertex $u_i$ with only one neighbour in $v$ in the strip cannot be a reflex vertex. Thus, the reflex vertex $u_i$ of the ring has at least two neighbours in $v$ in $T_1$.

    We will want to show that all corner cutters of $u_i$ in quadrilaterals $v_{j-1}v_{j}v_{j+1}u_i$ necessarily intersect $u_{i-1}u_i$ if the angle $\widehat{u_{i-1}u_i v_{j+1}}$, defined by the cone $u_{i-1}u_i v_{j+1}$ containing the vertices $v_s$ for $-1\le s\le j$, is smaller or equal to $\pi$. For that statement to be proven, we first need to show that the concatenation of these triangles form a convex polygon $u_i u_{i-1} v_{-1}\cdots v_{j+1}$ when the angle previously defined satisfies $\widehat{u_{i-1}u_i v_{j+1}}\le \pi$.

    \begin{figure}[tbhp]
        \centering
        \includegraphics[width=0.7\textwidth]{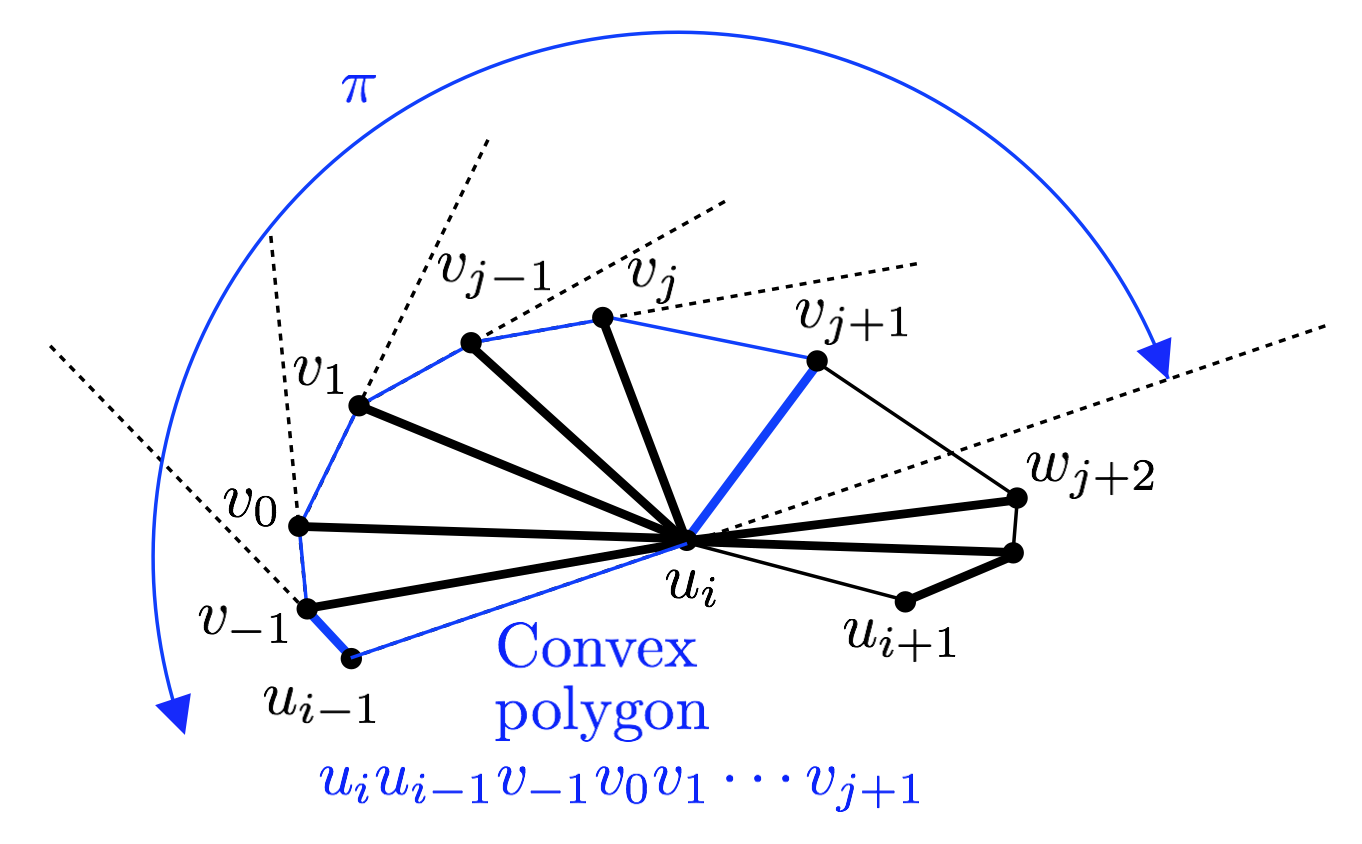}
        \caption{The successive concatenation of triangles of the ring around $u_i$ forms a convex quadrilateral $u_iu_{i-1}v_{-1}v_0\cdots v_{j+1}$ as long as the angle defined by the cone $u_{i-1}u_iv_{j+1}$ containing the other vertices of the polygon is smaller or equal to $\pi$. This result is due to the convexity of the successive overlapping quadrilaterals of the strip, which all have the common vertex $u_i$ in their diagonal.}
        \label{fig: lemma 3 reflex convex}
    \end{figure}

    We prove the convexity of the concatenation of triangles while the angle at $u_i$ of the agglomerated polygon is less than or equal to $\pi$ by a simple induction. See \cref{fig: lemma 3 reflex convex} for an illustration. The result easily holds for $u_i u_{i-1} v_{-1} v_0$, as by construction of the ring this quadrilateral is convex. Assume now that the polygon $u_i u_{i-1} v_{-1}\cdots v_{s}$ is convex for $s\ge 0$, and that the angle, defined by the cone $u_{i-1}u_iv_{s+1}$ and including $v_{-1}\cdots v_{s}$, is smaller or equal to $\pi$. By construction of the ring, $u_iv_{s-1}v_{s}v_{s+1}$ is a convex quadrilateral. This fact constrains $v_{s+1}$ to lie in the cone defined by $u_iv_{s-1}v_s$ and including the triangle $u_iv_{s-1}v_s$. However, due to the angular assumption on $v_{s+1}$ with respect to $u_{i-1}u_i$, necessarily $v_{s+1}$ also lies in the half-plane delimited by $u_{i-1}u_i$ and including the vertices $v_{-1}\cdots v_{s}$. Thus, $v_{s+1}$ belongs in the intersection of both domains. Then, all segments $u_{i-1}v_{s+1}$ and $v_{-1}v_{s+1}, \cdots, v_{s}v_{s+1}$ are included in the polygon $u_i u_{i-1} v_{-1}\cdots v_{s+1}$, which means that it is convex.

    \begin{figure}[tbhp]
        \centering
        \includegraphics[width=0.7\textwidth]{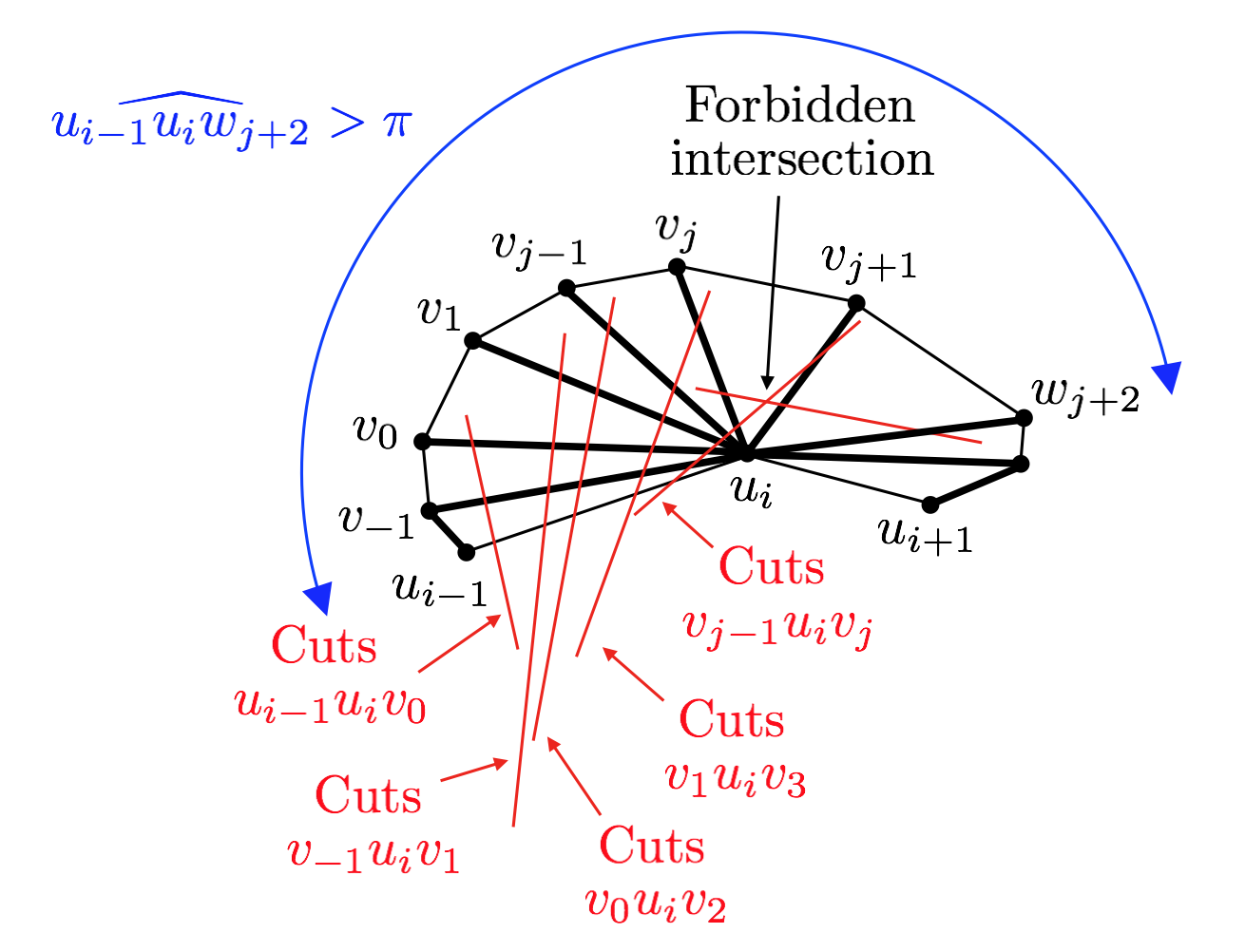}
        \caption{The corner cutter of $u_i$ in the quadrilateral $u_iu_{i-1}v_{-1}v_0$ necessarily intersects $u_{i-1}u_i$. Furthermore, the polygon of successively concatenated triangles is convex when the polygon is entirely contained within a half-plane delimited by $u_{i-1}u_i$. Thus, a simple induction shows that all of the corner cutters of $u_i$ inside the quadrilaterals of this convex polygon must intersect $u_{i-1}u_i$. However, the corner cutter of $u_i$ in the  first quadrilateral that cannot be concatenated with the previous convex polygon in a convex way will intersect the corner cutter of $u_i$ for the previous quadrilateral inside the triangulated ring. This statement leads to a final contradiction.}
        \label{fig: lemma 3 corner cutters}
    \end{figure}

    We now prove that, in each quadrilateral of such a concatenation of triangles $u_i u_{i-1} v_{-1}\cdots v_{j+1}$ with the previously defined angle $\widehat{u_{i-1}u_iv_{j+1}}\le \pi$, the corner cutters of $u_i$ intersect $u_{i-1}u_i$. See \cref{fig: lemma 3 corner cutters} for an illustration. First, notice that they cannot come from $u_{i-1}$ or $v_s$ for any $-1\le s\le j-1$, as such corner cutters would create edges intersecting quadrilaterals with maximal diagonals that come from one of the vertices of that same quadrilateral, which is a property we showed could never hold under the assumption that no flip of a maximally intersecting edge reduces the total number of intersections. Thus, the corner cutters of $u_i$ in $u_iu_{i-1}v_{-1}v_0$ intersects $u_{i-1}u_i$. Look now at the corner cutters of $u_i$ in the quadrilateral $u_i v_{-1}v_0v_1$. These edges of $T_2$ necessarily intersect either $u_{i-1}u_i$ or $u_{i-1}v_{-1}$, in particular they cannot come from $u_{i-1}$. However, by convexity of the polygon $u_iu_{i-1}v_0v_1$, if they intersect $u_{i-1}v_{-1}$, then they will also intersect the corner cutters of $u_i$ in the previous quadrilateral $u_i u_{i-1}v_{-1}v_0$, as these corner cutters intersect $u_{i-1}u_i$ and $u_{i}v_0$, and for which we know there exists at least one. Thus, the corner cutters of $u_i$ in $u_i v_{-1}v_0v_1$ intersect $u_{i-1}u_i$. We can repeat this argument of looking at corner cutters of $u_i$ along successive quadrilaterals by using the fact that the polygon $u_i u_{i-1} v_{-1}\cdots v_{s+1}$ remains convex for $-1\le s\le j$, which proves by induction that in each quadrilateral of a concatenation of triangles $u_i u_{i-1} v_{-1}\cdots v_{j+1}$ with the previously defined angle $\widehat{u_{i-1}u_iv_{j+1}}\le \pi$, the corner cutters of $u_i$ intersect $u_{i-1}u_i$.
    
    We are now ready to ready to reach a contradiction. See \cref{fig: lemma 3 corner cutters} for an illustration. Since $u_i$ is a reflex vertex of the ring, the previously defined angle $\widehat{u_{i-1}u_iu_{i+1}}$ along this cone containing the vertices $v_s$ for $-1\le s\le r$, is strictly greater than $\pi$. We can thus take the first vertex $w_{j+2}$ such that $\widehat{u_{i-1}u_iw_{j+2}} >\pi$, where $w_{j+2}$ is either $v_{j+2}$ for some $j\le r-2$ or either $u_{i+1}$, i.e. $\widehat{u_{i-1}u_iv_{j+1}}\le \pi$ and $\widehat{u_{i-1}u_iv_{j+2}}> \pi$ using the previous definition of angles. Since $\widehat{u_{i-1}u_iv_{j+1}}\le \pi$, the corner cutters of $u_i$ in the quadrilateral $u_i w_{j-1}v_jv_{j+1}$ intersect $u_{i-1}u_i$, where $w_{j-1}$ is $v_{j-1}$ if $j\ge 0$ and $u_{i-1}$ otherwise. On the other hand, the corner cutters of $u_i$ in $u_i v_j v_{j+1}w_{j+2}$ intersect $u_i w_{j+2}$ and either $v_j v_{j+1}$ or either $v_{j+1} w_{j+2}$. In particular, they cannot intersect $u_{i-1}u_i$ since $\widehat{u_{i-1}u_iv_{j+2}}> \pi$. Thus, they must intersect the previous corner cutters inside the polygon $u_i w_{j-1} v_{j}v_{j+1}w_{j+2}$, which is not allowed as it would create a vertex inside faces of the triangulation $T_1$ or intersecting edges in the planar triangulation $T_2$. However, these corner cutters must exist as previously proven. Therefore, we have reached a contradiction.
    
    Due to the fact that De Loera et al. \cite{de2010triangulations} missed an important case when assuming that only the first case was possible when extending the $3$-zigzag to the next quadrilateral (they incorrectly implicitly claim that $bt$ is maximal when $ad$, $ac$, and $bc$ are maximal), they seem to reach a contradiction much faster. Indeed, as a consequence of their incorrect claim, any vertex of the ring can only be present in two overlapping quadrilaterals, i.e. that any $u_i$ is connected to two vertices $v_j$ and $v_{j+1}$ but not $v_{j+2}$, and that likewise all $v_j$ are connected to two $u_i$ and $u_{i+1}$ but not $u_{i+2}$. This incorrect result leads to a simplified analysis of corner cutters at reflex vertices of the ring, say $u_i$ without loss of generality, as then the corner cutters of $u_i$ in the two successive quadrilaterals must necessarily meet inside the overlapping quadrilaterals which leads to a contradiction. To avoid the pitfall of De Loera et al. \cite{de2010triangulations}, we saw that further work was needed, as simply analysing two consecutive corner cutters around a reflex vertex does not necessarily provide a contradiction. We must instead find the right successive overlapping quadrilaterals for that reflex vertex that guarantees a contradiction when studying their corner cutters of the reflex vertex.

    This result is the final contradiction against the earliest assumption. This assumption was that no matter what maximally intersecting edge of $T_1$ with $T_2$ we chose, flipping it would not strictly reduce the number of intersections with $T_2$. Thus, we have proven that there exists a maximally intersecting edge of $T_1$ that we can flip, and for which this flip operation reduces the number of intersections with $T_2$ by at least one.
    
\end{proof}

We can now conclude the proof of \cref{th: goal}.

\begin{proof}
    Starting from $T_1$, as long as the current triangulation differs from $T_2$, there exists edges of the current triangulation with at least one intersection with $T_2$ that are also have maximal intersections with $T_2$. According to \cref{lemma 3}, there exists at least of these edges such that when flipped, the total number of intersections with $T_2$ reduces by at least one. To reach $T_2$ we need only reduce the number of intersections to $0$. By following the strategy of flipping a maximally intersecting edge with $T_2$ that reduces the number of intersections, we perform a sequence of flips that reduce this number by at least one at each step until it reaches $T_2$. Therefore, we found a sequence of edge flips bringing $T_1$ to $T_2$ using at most the original total number of intersections $\#(T_1,T_2)$ between $T_1$ and $T_2$. By taking the minimum length sequence of edge flips from $T_1$ to $T_2$, we have proven that ${d_f(T_1,T_2) \le \#(T_1,T_2)}$.
    
    The last inequality of \cref{th: goal} is achieved by naively upper-bounding the number of intersections of both triangulations using a worst case analysis. Since both triangulations share the same boundary edges, in the worst case, i.e. the one given the most intersections, all inner edges of $T_1$ intersect all inner edges of $T_2$. We then estimate the number of these edges using the Euler formula. The proof is well-known and thus not detailed in Hanke et al. \cite{hanke1996edge} nor in De Loera et al. \cite{de2010triangulations}. We nevertheless provide its details here for completeness. 
    
    For planar graphs with holes, the Euler formula is 
    \begin{equation}
        n-e+f=1-h,
    \end{equation}
    where $n$ is the number of vertices, $e$ is the number of edges, $f$ is the number of faces (excluding the outer face and the holes), and $h$ is the number of holes (i.e. the number of fixed non triangulated inner faces). For each face, we can count the number of edges it sees. Since all faces are triangles, this accounts to a vote of $3f$. By looking at it from the point of view of the edges, each inner edge is neighboured by exactly two faces and each border edge is neighboured by exactly one face. Thus we get an edge count of $2e_{int}+e_{b}$, where $e_{int}$ is the number of inner edges and $e_b$ is the number of boundary edges. Note that $e_b = n_b$ is 
    the sum of the length of the polygons $B_i$, and that it is larger than the number of inner boundary vertices with equality only when each inner boundary vertex is only present on a single boundary polygon. To see this, cyclically order the vertices in a clockwise fashion along the boundary for each boundary contour (around each hole and around the outer face). This naturally translates to a cyclical ordering of boundary edges where each edge can be associated uniquely with its counter clockwise vertex and reciprocally. We have therefore an edge count of $2e_{int}+n_b$. By equalising both counts, we get that $3f = 2e_{int}+b$. We then plug in this result in the Euler formula (multiplied on both sides by 3) to obtain $3n-3e+2e_{int}+n_b = 3 - 3h$. However, $e = e_{int}+e_b$. Therefore, $e_{int} = 3n-2n_b-3 - 3h$. In particular, note how this formula does not depend on the triangulation of the set of vertices. Therefore, we can crudely upper-bound the total number of intersections by $\#(T_1,T_2)\le(3n-2n_b-3-3h)^2$.
    
    Hanke et al. \cite{hanke1996edge} similarly estimated the number of intersections as $3n-2n_b-3$ as they only considered full triangulations of the convex hull of the points $S$. De Loera \cite{de2010triangulations} also use this assumption in their book and provide the same bound. However, as we have shown, the proof naturally extends to more general triangulations that can incorporate holes and an outer border only given by a simple polygon rather than simply the polygon defined as the boundary of the convex hull.

\end{proof}

\section{Conclusion}

We  revisited the proof that the flip distance between triangulations of a same finite set of vertices in a planar region is upper-bounded by the total number of edge intersections between them $d_f \le \#(T_1,T_2)$, and we provided a crude upper-bound estimate of this number of intersections. This upper-bound can be reached for specific configurations. On the other hand, there are known examples of optimal sequences of edges flips that do not necessarily follow the heuristic of systematically reducing the number of intersections \cite{hanke1996edge}. 

The global line of attack of the proof is due to Hanke et al. \cite{hanke1996edge}, later revised by De Loera et al. \cite{de2010triangulations} in their book. Our main contribution is to provide what we believe to be a fully detailed and extensively revised proof that corrects some errors and some false claims, while detailing every step in a hopefully readable fashion. Furthermore, due to our consideration of full details, we showed that the proof readily applies to triangulations of simple polygons and polygonal regions with holes and hence refined the estimation of the number of intersections to take into account this possibility.

\bibliographystyle{siamplain}
\bibliography{references}
\end{document}